\documentclass[11pt, letterpaper]{article}

\usepackage{nameref}
\usepackage[pagebackref,colorlinks,linkcolor=blue,filecolor = blue,citecolor = blue, urlcolor = blue, hyperfootnotes=false]{hyperref}
\usepackage[capitalize, noabbrev]{cleveref}

\usepackage{fullpage}
\usepackage{amsthm}
\usepackage{thmtools}
\usepackage{amsopn}
\usepackage{amsmath,amssymb,amsfonts,nicefrac}
\usepackage{xspace}
\usepackage{color}
\usepackage{url}
\usepackage{mathtools}
\usepackage{booktabs}
\usepackage{mdframed}
\usepackage{mdwlist}
\usepackage{bm}
\usepackage{bbm}
\usepackage{times}
\usepackage{BOONDOX-cal} 

\usepackage{todonotes}

\newcounter{mynotes}
\newcommand{\Wq}{\widetilde{Q}}

\newtheorem{theorem}{Theorem}[section]

\newtheorem{definition}[theorem]{Definition}

\newtheorem{lemma}[theorem]{Lemma}
\newtheorem{fact}[theorem]{Fact}

\newtheorem{cl}[theorem]{Claim}
\newtheorem{claim}[theorem]{Claim}

\newtheorem{remark}[theorem]{Remark}

\Crefname{theorem}{Theorem}{Theorems}
\Crefname{lemma}{Lemma}{Lemmas}

\newcommand{\junk}[1]{}
\newcommand{\ignore}[1]{}
\newcommand{\E}{\mathbb{E}}

\newcommand{\R}[0]{{\ensuremath{\mathbb{R}}}}
\newcommand{\N}[0]{{\ensuremath{\mathbb{N}}}}
\newcommand{\Z}[0]{{\ensuremath{\mathbb{Z}}}}
\newcommand{\Ft}{\ensuremath{\mathbb{F}[2]}}

\newcommand{\nfrac}[2]{\nicefrac{#1}{#2}}

\newcommand{\eps}{\varepsilon}

\newcommand{\LC}{\textsc{Smooth Label Cover}\xspace}
\newcommand{\NP}{\textsc{NP}\xspace}
\newcommand{\Pclass}{\textsc{P}\xspace}

\newcommand{\ol}[1]{\ensuremath{\overline{#1}}\xspace}
\newcommand{\wh}[1]{\ensuremath{\widehat{#1}}\xspace}
\newcommand{\mc}[1]{\ensuremath{\mathcal{#1}}\xspace}

\newcommand{\mbf}[1]{\ensuremath{{\bf #1}}\xspace}
\newcommand{\mcb}[1]{\ensuremath{\mathcalboondox{#1}}\xspace} 
\newcommand{\mb}[1]{\ensuremath{\mathbf{#1}}\xspace}
\newcommand{\tn}[1]{\ensuremath{\textnormal{#1}}\xspace}

\def\multiset#1#2{\ensuremath{\left(\kern-.3em\left(\genfrac{}{}{0pt}{}{#1}{#2}\right)\kern-.3em\right)}}

\newcommand{\slfrac}[2]{\left.#1\middle/#2\right.}

\newcommand{\initOneLiners}{%
    \setlength{\itemsep}{0pt}
    \setlength{\parsep }{0pt}
    \setlength{\topsep }{0pt}
}

\def\showauthornotes{1}
\ifnum\showauthornotes=1
\newcommand{\Authornote}[2]{{\sf\small\color{red}{[#1: #2]}}}
\else
\newcommand{\Authornote}[2]{}
\fi

\newcommand{\anote}[1]{}
\newcommand{\rnote}[1]{}
\newcommand{\snote}[1]{}

\title{Hardness of learning noisy halfspaces using polynomial thresholds}
\author{Arnab Bhattacharyya\thanks{Indian Institute of Science, Bangalore, India. Supported in part by DST Ramanujan Grant DSTO1358. Research partly done while visiting the Simons Institute for Theory of Computing, Berkeley, USA. Email: {\tt arnabb@csa.iisc.ernet.in}} \and Suprovat Ghoshal\thanks{Indian Institute of Science, Bangalore, India. Email: {\tt suprovat.ghoshal@csa.iisc.ernet.in}} \and Rishi Saket\thanks{IBM Research, Bangalore, India. Email: {\tt rissaket@in.ibm.com}}}

\date{}

\begin{document}
\maketitle

\begin{abstract}
	We prove the hardness of weakly learning halfspaces in the presence of adversarial noise using polynomial threshold functions (PTFs). In particular, we prove that for any constants $d \in \mathbb{Z}^+$ and $\eps > 0$, it is NP-hard to decide: given a set of $\{-1,1\}$-labeled points in $\mathbb{R}^n$  whether (YES Case) there exists a halfspace that classifies $(1-\eps)$-fraction of the points correctly, or (NO Case) any degree-$d$ PTF classifies at most $(1/2 + \eps)$-fraction of the points correctly. This strengthens to all constant degrees the previous NP-hardness of learning using degree-$2$ PTFs shown by Diakonikolas et al. (2011). The latter result had remained the only progress over the works of Feldman et al. (2006) and Guruswami et al. (2006) ruling out weakly proper learning adversarially noisy halfspaces.

\end{abstract}

\section{Introduction}

Given a distribution $\mc{D}$ over $\{-1,1\}$-labeled points in
$\mathbb{R}^n$, the accuracy of a classifier function $f : \mathbb{R}^n \to
\{-1,1\}$ is the probability that $f(x) = \ell$ for a random point-label pair $(x, \ell)$
sampled from $\mc{D}$. 
A concept class $\mc{C}$ is said to be \emph{learnable}
by hypothesis class $\mc{H}$ if there is an efficient procedure which,
given access to samples from any distribution $\mc{D}$ consistent with some $f \in
\mc{C}$, generates with high probability a classifier $h \in \mc{H}$ of
accuracy approaching that of $f$ for $\mc{D}$. When $\mc{H}$ can be
taken as $\mc{C}$ itself, the latter is said to be \emph{properly} learnable. The
focus of this work is one
of the simplest and most well-studied concept classes: the \emph{halfspace} which
maps $x \in \mathbb{R}^n$ to $\tn{sign}(\langle v, x\rangle - c)$ for
some $v \in \mathbb{R}^n$ and $c \in \mathbb{R}$. The study of
halfspaces goes back several decades to the development of various
algorithms in artificial intelligence and machine learning such as the
Perceptron~\cite{Rosenblatt62,MP69} and  SVM~\cite{CV95}.  Since then,
halfspace-based classification has found applications in many other
areas, such as computer vision~\cite{OJM90} and
data-mining~\cite{RRK04}.

It is known that a halfspace can be properly learnt by using linear
programming along with a polynomial number of samples to compute a
separating hyperplane~\cite{BEHW89}. In noisy data however, it is not
always possible to find a hyperplane separating the differently
labeled points. Indeed, in the presence of (adversarial) noise, 
i.e.  the \emph{agnostic} setting, proper learning of a halfspace to
optimal accuracy with no distributional assumptions  
was shown to be NP-hard by Johnson and
Preparata~\cite{JP78}.  Subsequent results showed the hardness of
approximating the accuracy of properly learning a noisy halfspace to
constant factors: $\tfrac{262}{261}-\eps$ by Amaldi and
Kann~\cite{AK98}, $\tfrac{418}{415}-\eps$ by Ben-David et
al.~\cite{BDEL00}, and $\tfrac{85}{84}-\eps$ by Bshouty and
Burroughs~\cite{BB06}. These results were considerably strengthened 
independently by Feldman,
Gopalan, Khot, and Ponnuswami~\cite{FGKP09} and by Guruswami and
Raghavendra~\cite{GR09}\footnote{The reduction of Guruswami and Raghavendra~\cite{GR09}
works even for the special case when the points are over the boolean hypercube.} 
who proved  hardness of even \emph{weakly}
proper learning a noisy halfspace, i.e. 
to an accuracy 
beyond the random threshold of $1/2$. This implies an optimal
$(2-\eps)$-inapproximability in terms of the learning accuracy. 
Building upon these works Feldman, Guruswami,
Raghavendra, and Wu~\cite{FGRW12} showed that the same hardness holds 
for learning noisy monomials (OR functions over the boolean hypercube) using
halfspaces.

At this point, it is  natural to ask whether the halfspace learning problem remains hard
if the classifier is allowed to be from a larger class of functions, i.e., {\em non-proper} learning.
In particular, consider the class of degree-$d$ \emph{polynomial
threshold functions} (PTF) which are given by mapping $x \in
\mathbb{R}^n$ to $\tn{sign}(P(x))$ where $P$ is a degree-$d$ polynomial.
They generalize halfspaces a.k.a.~\emph{linear threshold functions}
(LTFs) which are degree-$1$ PTFs and are very common hypotheses in machine learning
because they are output by kernelized models (e.g., perceptrons, SVM's, kernel k-means, kernel PCA, etc.) when instantiated with the polynomial kernel. From a complexity viewpoint,
PTFs were studied by   Diakonikolas, O'Donnell, Servedio, and
Wu~\cite{DOSW11} who showed the hardness of weakly  proper 
learning a noisy degree-$d$ PTF for any constant $d \in
\mathbb{Z}^+$, assuming Khot's Unique Games Conjecture
(UGC)~\cite{Khot02}. On the other hand, proving the hardness of
weakly learning noisy halfspaces using degree-$d$ PTFs has turned out
to be quite challenging. Indeed, the only such result is by
Diakonikolas et al.~\cite{DOSW11} who showed the corresponding hardness 
of learning using a degree-$2$ PTF. With no further
progress till now, the situation remained unsatisfactory. 

In this work, we significantly advance our understanding  by proving  the hardness 
of weakly learning an $\eps$-noisy halfspace by a degree-$d$ PTF \emph{for any
constant} $d \in \mathbb{Z}^+$. Our main
result is formally stated as follows.

\begin{theorem}\tn{(This work)}\label{thm:main} For any constants $\delta > 0$, and $d
\in \mathbb{Z}^+$, it is \NP-hard to decide whether a given
set of $\{-1,1\}$-labeled points in $\mathbb{R}^n$ satisfies:

\smallskip
\noindent
~~~~~~~\tn{\textbf{YES Case.}} There exists a halfspace that correctly classifies
$(1-\delta)$-fraction of the points, or

\smallskip
\noindent
~~~~~~~\tn{\textbf{NO Case.}} Any degree-$d$ PTF
classifies at most $(1/2+\delta)$-fraction of the points
correctly.

\smallskip
\noindent
The \tn{NO} case
 can be strengthened to rule out any function of constantly many
degree-$d$ PTFs.
\end{theorem}

To place our results in context, we note that algorithmic results for
learning noisy halfspaces are known under assumptions on the
distribution of the noise or the pointset.
In the presence of \emph{random
classification noise}, Blum, Frieze, Kannan, and Vempala~\cite{BFKV96}
gave an efficient learning algorithm approaching optimal accuracy, 
which was improved by Cohen~\cite{Cohen97}
who showed that in this case the halfspace can in fact 
be properly learnt. For
certain well behaved distributions, Kalai, Klivans, Mansour, and
Servedio~\cite{KKMS05} showed that halfspaces can be learnt even in
the presence of adversarial noise. Subsequent works by 
Klivans, Long, and Servedio~\cite{KLS09},
and Awasthi, Balcan, and Long~\cite{ABL17} improved the noise
tolerance and introduced new algorithmic techniques. Building upon
them, Daniely~\cite{Daniely15} recently obtained a PTAS for minimizing
the hypothesis error with respect to the uniform distribution over a sphere.
Several of these learning algorithms
use halfspaces and low degree PTFs 
(or simple combinations thereof) as their hypotheses, 
and one could conceivably apply their techniques to the setting 
 without any distributional
assumptions. Our work provides evidence to the contrary.

\subsection{Previous related work}

Hypothesis-independent intractability results for 
learning for halfspaces are also known, but they make average-case or
cryptographic hardness assumptions which seem considerably stronger than
\Pclass$\neq$\NP.  Specifically, for exactly learning noisy halfspaces,
such results have been shown in the works of Feldman et
al.~\cite{FGKP09}, Kalai et al.~\cite{KKMS05}, Kothari and
Klivans~\cite{KK14}, and Daniely and Shalev-Shwartz~\cite{DS16}. In a
recent work, Daniely~\cite{Daniely16} rules out weakly learning
noisy halfspaces assuming the intractability of strongly refuting random
$K$-XOR formulas. On the other hand, Applebaum, Barak, and
Xiao~\cite{ABX08} have shown that hypothesis-independent
hardness results under standard complexity assumptions would imply a
major leap in our current understanding of complexity theory and are
unlikely to be obtained for the time being. Therefore, any 
study (such as ours) of the standard complexity-theoretic hardness of learning
halfspaces would probably need to constrain the hypothesis.

A natural generalization of the learning halfspaces problem is that of
learning intersections of two or more halfspaces.  Observe that unlike
the single halfspace, properly learning the intersection of two
halfspaces without noise does not in general admit a
separating hyperplane based solution.  Indeed, this problem was shown
to be NP-hard by Blum and Rivest~\cite{BR93}, later strengthened by
Alekhnovich, Braverman, Feldman, Klivans, and Pitassi~\cite{ABFKP08} to
rule out intersections of constantly many halfspaces as hypotheses.
The corresponding hardness of even weak learning was established by
Khot and Saket~\cite{KS11}, while Klivans and Sherstov~\cite{KS09}
proved under a cryptographic hardness assumption the intractability of
learning the intersection of $n^\eps$ halfspaces. Algorithms
for learning intersections of constantly many halfspaces have been
given in the works of Blum and Kannan~\cite{BK97} and
Vempala~\cite{Vempala97} for the uniform distribution over the unit
ball, Klivans, O'Donnell, and Servedio~\cite{KOS04} for the uniform
distribution over the boolean hypercube, and by Arriaga and
Vempala~\cite{AV06} and Klivans and Servedio~\cite{KlivansS08} for instances
with good \emph{margin}, i.e. the points being well separated from the
hyperplanes.

As was the case for learning a single noisy halfspace, there is
no known NP-hardness for learning intersections of two halfspaces
using (intersections of) degree-$d$ PTFs. This cannot, however, be
said of the finite field analog of learning halfspaces, i.e. the
problem of learning noisy parities over $\Ft$. While
H\r{a}stad's~\cite{Hastad} seminal work itself rules out
weakly proper learning a noisy parity over $\Ft$, later work of
Gopalan, Khot, and Saket~\cite{GKS10} showed the hardness of
learning an $\eps$-noisy parity by a degree-$d$ PTF to
within $(1 - 1/2^d + \eps)$-accuracy  -- which, however, 
is not optimal for $d > 1$.
Shortly thereafter, Khot~\cite{Khot-personal} 
observed\footnote{Khot's observations
remained unpublished for while, before they were included with his
permission by Bhattacharyya et al.~\cite{BGGS16} in their paper which
made a similar use of Viola's~\cite{Viola09} pseudo-random generator.}
that Viola's~\cite{Viola09} pseudo-random generator fooling degree-$d$
PTFs can be combined with coding-theoretic inapproximability 
results to yield
optimal lower bounds for all constant degrees $d$. From the
algorithmic perspective, one can learn an $\eps$-noisy parity over the
uniform distribution in $2^{O(n/\log n)}$-time as shown by Feldman et
al.~\cite{FGKP09} and Blum et al.~\cite{BKW03}. For general
distributions, Kalai, Mansour, and Verbin~\cite{KMV08} gave a
non-proper $2^{O(n/\log n)}$-time algorithm achieving an accuracy
close to optimal. 

Several of the inapproximability results mentioned above, e.g.~those of \cite{GR09}, \cite{GKS10}, \cite{KS11}, \cite{FGRW12} and \cite{DOSW11},
follow the \emph{probabilistically checkable proof
(PCP) test} based approach for
their hardness reductions. While our result builds upon these methods,
in the remainder of this section, we give an overview of our techniques
and describe the key enhancements which allow us to overcome some of
the technical limitations of previous hardness reductions.

\subsection{Overview of Techniques}

For hardness reductions, due to the uniform convergence results of \cite{Hau92, KSS94}, it is sufficient to take the optimization 
version of the learning halfspaces problem which consists of a set of
coordinates and a finite set of labeled points, the latter 
replacing a random distribution. A typical reduction (including ours)
given a hard instance of a constraint satisfaction problem (CSP) $\mc{L}$ over
vertex set $V$ and label set $[k]$, defines $\mc{C} := V\times [k]$
to be the set of coordinates over $\mathbb{R}$.
We let the formal variables
$Y_{(w,i)}$ be associated with the coordinate $(w,i) \in \mc{C}$.
The hypothesis $H$ (the
\emph{proof} in PCP terminology) is defined
over these variables. In our case, the proof will be a degree-$d$ PTF. The PCP test chooses
randomly a small set of vertices $S$
of $\mc{L}$, and runs a \emph{dictatorship} test on $S$: it tests $H$ on a set
of labeled points $P_S \subseteq \mathbb{R}^{\mc{C}}$ generated by the
dictatorship test. We desire the following two properties from the test: 
\begin{itemize*}
\item \textbf{(completeness)} if $H$
``encodes'' a
good labeling for $S$, then it is a good classifier for  $P_S$,
\item \textbf{(soundness)} a good classifier $H$ for
$P_S$ can be ``decoded'' into a good labeling for $S$. 
\end{itemize*}
 The soundness
property is leveraged to show that if $H$ classifies
$P_S$ for a significant fraction of the choices $S$, it can be used to
define a good global labeling for $\mc{L}$.
The CSP of choice in the above template is usually the Label Cover or
the Unique Games problem. While the NP-hardness of Label Cover is
unconditional, its projective constraints seem to present 
technical roadblocks --
also faced by Diakonikolas et al.~\cite{DOSW11} -- in analyzing 
learnability by  degree-$d$  ($d > 2$) PTFs.

Our work overcomes these issues and gives a hardness reduction from
Label Cover. The key ingredient to incorporate the Label Cover
projective constraints is a \emph{folding} over an appropriate
subspace defined by them. This amounts to restricting the entire instance
to the corresponding orthogonal subspace.  Similar folding for analyzing linear forms
has been used earlier in the works of Khot and Saket~\cite{KS11},
Feldman, Guruswami, Raghavendra, and Wu~\cite{FGRW12}, and Guruswami,
Raghavendra, Saket, and Wu~\cite{GRSW}. We are able to extend it over
degree-$d$ polynomials leveraging the linear-like structure decoded by
an appropriate dictatorship test. This uses a \emph{smoothness} property of 
the constraints (analogous to \cite{KS11,FGRW12,GRSW})  
of the Label Cover instance which is combined with the 
dictatorship test -- along with folding --
to yield the PCP test.

In the rest of this section, we informally describe our dictatorship
test, the motivation behind its design and the key ingredients
involved in its analysis. To begin, we present a simple preliminary
dictatorship test $\mathfrak{P}_0$ over $\mathbb{R}^{k}$ which works for linear thresholds. 
Of course, the NP-hardness of properly learning noisy halfspaces is
already known~\cite{FGKP09,GR09}, so this test does not yield anything
new.
Our purpose is illustrative and we include a sketch
of the arguments of its analysis. 
Taking $\eps > 0$ as a
small constant and  $\eta > 0$ a small parameter (to be defined
later), the description of $\mathfrak{P}_0$ is given in Figure \ref{fig:pcptest-0-intro}.
\begin{figure}
\begin{mdframed}
\begin{center} $\mathfrak{P}_0(\mathbb{R}^k, \eta, \eps)$ tests halfspace
$\tn{sign}(f(Y))$. \end{center}
\begin{enumerate*}

\item Sample $b \in \{-1,1\}$ uniformly at random.

\item Choose a random ``noise'' subset $\mc{I} \subseteq [k]$ by
including each $i$ independently with probability $\eps$.

\item For $i \in [k]\setminus\mc{I}$,  
set
$y_i = b\eta$,

\item For $i \in \mc{I}$, sample $y_i$ 
i.i.d. at random from $N(0,1)$.

\item Accept iff $\tn{sign}(f(y)) =  b$.
\end{enumerate*}
\end{mdframed}
\caption{Dictatorship Test $\mathfrak{P}_0$}
\label{fig:pcptest-0-intro}
\end{figure}

Observe that the linear threshold $\tn{sign}(Y_i)$ for each $i \in
[k]$ correctly classifies $(y,b)$ with probability $(1-\eps)$. 
In other words, every \emph{dictator} corresponds to a good
solution.

\subsubsection{Soundness analysis of
$\mathfrak{P}_0$}\label{sec-intro-soundness-P0}
Suppose there exists a linear form $f =  
\sum_{i \in [k]}\wh{f}_iY_i$ (assuming for
simplicity $f$ has no
constant term) such that
$\tn{sign}(f)$ passes $\mathfrak{P}_0$
 with probability $1/2 +
2\xi$ for
some $\xi = \Omega(1)$. Using (by now) standard analytical
arguments, we show that there exists $i^* \in [k]$ such that
\begin{equation}
\wh{f}_{i^*}^2 \geq
\Omega(1)\cdot\sum_{i\in[k]}\wh{f}_{i}^2 > 0. 
\label{eqn-into-example-bound}
\end{equation}
In other words, every good solution $f$ can be decoded into a 
dictator.

\medskip
It is not particularly challenging to obtain \eqref{eqn-into-example-bound}.
However, we sketch a systematic proof which shall be useful
when analyzing a more complicated dictatorship test for PTFs.  

Call a setting of $\mc{I}$ \emph{good} if $\tn{sign}(f)$ passes the test
conditioned on $\mc{I}$ with probability $1/2 + \xi$. By averaging,
it is easy to see that $\Pr_{\mc{I}}\left[\mc{I}\tn{ is good}\right] \geq
\xi/2$.  Let us fix such a good
$\mc{I}$. Without loss of generality, we may assume that $\mc{I} = \{k^* + 1, \dots, k\}$ and
further that $k^* \geq k/2$ by the Chernoff bound. We now define 
 $\{W_1, \dots, W_{k^*}\}$ as a basis for $\{Y_{i}\,\mid\,i \in [k^*]\}$
where $W_1 := (1/k^*)\sum_{i\in[k^*]} Y_{i}$, 
such that $\{W_1, \dots, W_{k^*}\}$ is an orthogonal
transformation of $\{Y_{i}\,\mid\,i \in [k^*]\}$ of the same $1/\sqrt{k^*}$ 
norm. 
Thus, we may rewrite $f$ as:
\begin{equation}\label{eqn-fYW-intro}
f = \sum_{i \in
[k]\setminus[k^*]}\tilde{f}_{i}Y_{i} + \sum_{\ell \in [k^*]}
\ol{f}_{\ell}W_\ell.
\end{equation} The variables in the first sum in the RHS of the
above are all i.i.d. $N(0,1)$. Further, it can be seen that under the
test distribution, $W_1 = b\eta$, and $W_\ell = 0$ ($\ell
= 2,\dots, k^*$). Therefore, we may assume that,
\begin{equation}
\ol{f}_1^2 > 0.\label{eqn-mass-nonzero-intro}
\end{equation}
Since the sign of $f$ must flip with that of $b$
with probability $\Omega(\xi) = \Omega(1)$, one can apply
Carbery-Wright's Gaussian
anti-concentration theorem to show that,
\begin{equation}\label{eqn-apply-anti-conc-intro}
\sum_{i \in
[k]\setminus[k^*]}\tilde{f}_{i}^2 \leq O(\eta^2)\ol{f}^2_{1}, 
\end{equation}
since otherwise, contributions from the first sum of \eqref{eqn-fYW-intro}  will 
overwhelm the contribution of $W_1$ to $f$.
Further, from the definition of $\{W_\ell\}_{\ell=1}^{k^*}$, we obtain
\begin{equation}
\sum_{i \in [k^*]}\tilde{f}_{i1}^2 = \frac{1}{k^*}\sum_{\ell \in [k^*]}
\ol{f}_{\ell}^2 \geq  \ol{f}_1^2/k^*. \label{eqn-unwinding-intro}
\end{equation}
Let us now revert to the notation with $\mc{I} = [k]\setminus[k^*]$. 
Using   \eqref{eqn-unwinding-intro} along with   \eqref{eqn-apply-anti-conc-intro},
and taking $\eta = o(\eps^3/\sqrt{k})$ one can ensure that,
\begin{equation}
\sum_{i \in \mc{I}}\tilde{f}_{i}^2 \leq
\frac{\eps}{10}\sum_{i \in [k]}\tilde{f}_{i}^2,
\label{eqn-1-bound-intro}
\end{equation}
and from   \eqref{eqn-mass-nonzero-intro} we obtain
\begin{equation}
\sum_{i \in [k]}\tilde{f}_{i}^2 > 0. 
\label{eqn-3-bound-intro}
\end{equation}
Note that   \eqref{eqn-1-bound-intro} holds for every good
$\mc{I}$ which is at least $\xi/2$ fraction of the choices of
$\mc{I}$. Randomizing over $\mc{I}$, an application of the Chernoff-Hoeffding
bound shows that \eqref{eqn-1-bound-intro} holds only with substantially smaller
probability unless there exists $i^* \in [k]$ such that:
\begin{equation}
\tilde{f}_{i^*}^2 \geq \frac{\eps^3}{8}\sum_{i \in
[k]}\tilde{f}_{i}^2. \label{eqn-heavy-i-intro}
\end{equation}
The desired bound in   \eqref{eqn-into-example-bound} now easily follow
from \eqref{eqn-3-bound-intro} and
\eqref{eqn-heavy-i-intro}.
The details are omitted.

The main idea of the above 
methodical analysis is a natural definition of the $W$
variables using which we isolate the sign-perturbation $b\eta$ into a
single variable $W_1$! Gaussian anti-concentration directly lower bounds 
the squared mass corresponding to $W_1$. Moreover, when transforming back 
to the squared mass of $Y_{i}$ ($i \in [k]\setminus\mc{I}$), the presence
of the heretofore ignored $W_{\ell}$ ($\ell > 1$) terms can only increase this quantity, as shown
in \eqref{eqn-unwinding-intro}. Lastly, the 
the ``decoding list size''  does not depend on the
sign-perturbation parameter $\eta$ which can be taken to
be small enough to makes sure that this size is a constant
depending only on the noise parameter $\eps$ and
the marginal acceptance probability $\xi$ of the test.

\subsubsection{Enhancing the Dictatorship Test for degree-$d$ PTFs}

Our goal is a reduction proving the hardness of weakly
learning noisy halfspaces using degree-$d$ PTFs. One could hope
to utilize the dictatorship test $\mathfrak{P}_0$ itself for this purpose.
Unfortunately, this presents problems even for $d = 5$. To see this
consider  
the
degree-$5$ polynomial,
$$f(Y) = Y_{i^*}^3\left(\sum_{i\in[k]\setminus\{i^*\}}Y^2_i\right),$$
for some distinguished $i^* \in [k]$.
It is easy to see that $\tn{sign}(f)$ passes the test with probability close to
$1$. However, the distinguished variable $Y_{i^*}$ appears with a cubic
power in $f$, whereas the folding approach works well only when $Y_{i^*}$
occurs as a linear factor of some sub-polynomial. This is due to the inherently
linear nature of the folding constraints. Consequently, 
when $\mathfrak{P}_0$ is combined with a Label Cover
instance the analysis becomes infeasible. 

Our approach to overcome this bottleneck is for the PCP to test
several independently and randomly chosen vertices. For this, the
dictatorship test would be on the domain
$\mathbb{R}^{[k]\times[T]}$ where
$T$ is chosen much larger than the degree $d$ of the PTF to be tested.
The space $\mathbb{R}^{[k]\times[T]}$ is thought of as real space
spanned by $T$ blocks of $k$ dimensions each. 
In this case, if the test passes with probability $ > 1/2$, then there
is a way to decode a good label to at least one out of the $T$
blocks. 
A key step in our analysis crucially leverages the choice of $T$ to extract
out a specific sub-polynomial which is linear in the variables of one of the
$T$ blocks. This
is done via an application of the following lemma which is
proved in Section \ref{sec:struct}.
\begin{lemma}\label{lem-approx-structural-intro}
Given a degree-$d$ polynomial of the form $(Y_1 + \dots +
Y_T) \cdot S(Y_1,\dots, Y_T)$, where $T  > 2d$ and $S$ is a degree-$(d-1)$ polynomial, there exist at least $T/2$ indices $j \in [T]$
such that: for each such $j$, the sum of squares of the coefficients corresponding to the
terms (in the monomial representation) linear in $Y_j$ is
at least $c$ times the sum of squares of coefficients of $S$,
where $c := c(T,d) > 0$.  
\end{lemma}

In Figure \ref{fig:pcptest-1-intro}, 
we give a formal description of the Dictatorship test $\mathfrak{P}_1$
employed by our reduction. Its analysis builds upon that of
$\mathfrak{P}_0$ above, so we provide a short sketch.
Let $T = 10d$ and $\eps > 0$ be a constant, and $\eta > 0$ be
parameter to be defined later.
\begin{figure}
\begin{mdframed}
\begin{center}PCP Test $\mathfrak{P}_1(\mathbb{R}^{[k]\times[T]}, \eta, \eps)$ tests
degree-$d$ PTF
$\tn{sign}(P(Y))$ \end{center}

\begin{enumerate*}

\item Sample $\{\delta_{j}\, \mid\, j\in[T]\}$ from the joint Gaussian
	distribution where the marginals are $N(0,1)$,
	$\E[\delta_j\delta_{j'}] = -1/(T-1)$ for all $j\neq j'$, 
	and $\sum_{j=1}^T \delta_{j} = 0$.

\item Sample $b \in \{-1,1\}$ uniformly at random.

\item Sample $\mc{I} \subseteq [k]\times[T]$ to be a random subset
where each $(i,j) \in [k]\times[T]$ is added to $\mc{I}$
independently with probability $\eps$.

\item For each $(i,j) \in ([k]\times[T])\setminus \mc{I}$, set
$y_{ij} = (\sqrt{(T-1)/T})\delta_{j} +
b\eta$.

\item Independently for each $(i,j) \in \mc{I}$,
sample  $y_{ij} \sim N(0,{1})$.

\item Accept iff $\tn{sign}(P(y)) = b$.

\end{enumerate*}
\end{mdframed}
\caption{Dictatorship Test $\mathfrak{P}_1$}
\label{fig:pcptest-1-intro}
\end{figure}
Consider the linear
threshold given by,
$$\tn{sign}\left(\sum_{j=1}^TY_{i_jj}\right),$$
for any $i_j \in [k]$ ($1\leq j \leq T$). It is easy to see that this
passes the test with probability at least $(1 - \eps T)$. Thus,
choosing a dictator for each block yields a good solution for the
test.

For the soundness analysis, as in Section
\ref{sec-intro-soundness-P0} we fix a good noise set $\mc{I}$ 
conditioned on which the test accepts $P$ with probability at least
$1/2 + \xi$, and $\Pr[\mc{I} \tn{ is good }] \geq \xi/2.$ 
Further, without loss of generality, we assume that $\mc{I} = \cup_{j=1}^T \left(\{k_j + 1,\dots
k\}\times\{j\}\right)$, where (by Chernoff bound) $k_j \geq k/2$ for
$1\leq j \leq T$. For each $j$, $\{W_{1j},\dots, W_{k_jj}\}$ is
defined to be an orthogonal transformation of $\{Y_{1j},\dots,
Y_{k_jj}\}$ of the same $1/\sqrt{k_j}$ norm, where $W_{1j} =
(1/k_j)\sum_{i=1}^{k_j}Y_{ij}$. It is easy to see that $W_{1j} =  
(\sqrt{(T-1)/T})\delta_{j} + b\eta$, while $W_{\ell j} = 0$ under the
test distribution for $\ell > 1$.

Additionally, we also define $\{U_1,\dots, U_T\}$ to be an orthonormal
transformation of $\{W_{11},\dots, W_{1T}\}$ where $U_1 =
(1/\sqrt{T})\sum_{j=1}W_{1j}$. Again, it can observed that $U_1 =
(\sqrt{T})b\eta$ and $U_2,\dots, U_T$ are independent $N(0,1)$. Using
this we write  the polynomial $P = P' + Q_0 + U_1 Q_1$, where $P'$
consists of all the terms which have any $W_{\ell j}$, $\ell > 1$ as a
factor. Further, $Q_0$ is independent of $U_1$. Since $P' = 0$ under
the distribution we ignore it for now, noting that $\|Q_1\|_2^2 =
\E[Q_1^2] > 0$, since the test accepts with probability $ > 1/2$. 
The first step is to show, via Gaussian
anti-concentration on $Q_0$ and Chebyshev's inequality on $Q_1$, that
\begin{equation}\label{eqn:relative}
\|Q_0\|_2^2 \leq O(\eta^2)\|Q_1\|_2^2.
\end{equation}
Let us write $Q_1 = \sum_{H\in \mcb{H}} H\cdot Q_{1,H}(U_1,\dots, U_T)$, where the
sum is over the set $\mcb{H}$ of normalized Hermite monomials\footnote{By Hermite \emph{monomials},
we mean elements of the polynomial Hermite basis over the corresponding variables.} 
over the independent $N(0,1)$
variables $\cup_{j=1}^T\{Y_{ij}\}_{i=k_j+
1}^k$. Moreover, let  $Q^{(D)}_1 = \sum_{H \in \mcb{H}_D} 
H\cdot Q_{1,H}(U_1,\dots, U_T)$ for $0\leq D \leq
d-1 \geq \tn{deg}(Q_1)$, where $\mcb{H}_D$ is the subset of $\mcb{H}$
of degree exactly $D$.
Thus,  $\|Q_1\|_2^2 = \sum_{H\in \mcb{H}} \|Q_{1,H}\|_2^2$. Writing
$Q_{1,H} = Q_{1,H}(W_{11}, \dots, W_{1T})$ we also
define $\|Q_{1,H}\|_{\tn{mon}}^2$ as sum of squares of the
coefficients in 
the standard monomial basis $\mcb{M}$ of $\{W_{11}, \dots, W_{1T}\}$. A
straightforward calculation shows that:
\begin{equation}\label{eqn:relative2}
\|Q_{1,H}\|_2^2  \leq O(1)
\|Q_{1,H}\|_{\tn{mon}}^2,
\end{equation}
where the constants depending on $T$ and $d$ are absorbed in 
the $O(1)$ notation.	
On the other hand, since $Q_0$ is independent of $U_1$, using similar
definition of $Q_{0,H}$, we can
establish the reverse bound for it:
\begin{equation}\label{eqn:relative3}
\|Q_{0,H}\|_{\tn{mon}}^2  \leq O(1)\|Q_{0,H}\|_2^2.
\end{equation}

The rest of the arguments significantly build upon those in Section
\ref{sec-intro-soundness-P0}. We present a semi-formal description,
omitting much of the technical details. For
 reasons made clear later, we first 
carefully select $d^* \in \{0,\dots, d-1\}$ to be
the largest $D \in \{0,\dots, d-1\}$ such that $\|Q_1^{(D)}\|_2^2 \geq
\frac14\rho^D\|Q_1\|_2^2$ for a small enough constant depending on $k,T,
d,$ and $\eps$. It is easily observed that such a $d^*$ must exist
satisfying the properties: (i) $\|Q_1^{(d^*+1)}\|_2^2 \leq \frac14\rho^{d^* + 1} 
\|Q_1\|_2^2$, and (ii) $\|Q_1^{(d^*)}\|_2^2 \geq \frac14\rho^{d^*}\|Q_1\|_2^2$. 

Now we focus our attention on $U_1Q_1^{(d^*)}$ writing it as
\begin{equation}\label{eqn-U1Q1-intro}
U_1Q_1^{(d^*)} = \sum_{H \in \mcb{H}_{d^*}} HU_1Q_{1,H}(W_{11},\dots,
W_{1T}) = \sum_{H\in \mcb{H}_{d^*}}\sum_{M \in \mcb{M}}c_{H,M}HM.
\end{equation}
Let $\mcb{H}_{-j^*D}\subseteq \mcb{H}_{D}$ (resp.
$\mcb{M}_{-j^*} \subseteq \mcb{M}$) be the subset of basis elements
not containing any
variable from the $j^*$th block, i.e. $\{Y_{ij^*}\}_{k_{j^*}<i\leq k}$
(resp. $W_{1j^*}$). 
Now with $U_1 = (1/\sqrt{T})\sum_{j=1}W_{1j}$, we 
apply Lemma  \ref{lem-approx-structural-intro}
to each $U_1Q_{1,H}(W_{11},\dots, W_{1T})$ in the first expansion of
\eqref{eqn-U1Q1-intro}.
Using the fact that each $H$ has at most $d$
variables along with our choice of $T = 10d$ yields a $j^* \in [T]$
such that 
\begin{eqnarray}
\sum_{H \in \mcb{H}_{-j^*d^*}}\sum_{M \in \mcb{M}_{-j^*}}
c^2_{H,MW_{1j^*}} & \geq & \Omega(1)\bigg(\sum_{H \in \mcb{H}_{d^*}}\sum_{M
\in \mcb{M}} c^2_{H,M}\bigg) \label{eqn-lemma-apply-intro-1}\\
& \geq & \Omega(1)\|Q_1^{(d^*)}\|_2^2\ \geq\ \Omega(1)\rho^{d^*} \|Q_1\|_2^2
\label{eqn-lemma-apply-intro-2}
\end{eqnarray}
where the last two inequalities use
\eqref{eqn:relative2} along
with property (ii) above.

The next component of the analysis is to relate the bounds above with
the coefficients of
a suitable sub-polynomial of $P$ which is linear in the
variables $Y_{ij^*}$, $1 \leq i \leq k_{j^*}$. For this, let us first define
$\tilde{Q}$ to be exactly the sub-polynomial of $P$ 
 which does not contain any term with $W_{ij}$ where $i \neq
1$ and $j \neq j^*$. Rewriting the variables $\{W_{ij^*}\,\mid\, i \in
[k_{j^*}]\}$ in terms of $\{Y_{ij^*}\,\mid\, i \in [k_{j^*}]\}$, we
consider the sub-polynomial $\tilde{Q}_{\tn{lin}}$ (of $\tilde{Q}$) which is linear in the
variables $\{Y_{ij^*}\,\mid\, 1\leq i \leq k\}$. Note that
$\left(\cup_{D=0}^{d-1}\mcb{H}_{-j^*D}\right)\circ \mcb{M}_{-j^*}\circ\{Y_{ij^*}\}_{i=1}^k$ is a
basis in which $\tilde{Q}_{\tn{lin}}$ can be written with 
coefficients $\tilde{c}_{H, M, i}$ corresponding to the basis element
$HMY_{ij^*}$. Using the orthonormal transformation between
$\{W_{ij^*}\}_{i\in [k_{j^*}]}$ and $\{Y_{ij^*}\}_{i=1}^{k_{j^*}}$ we obtain
\begin{equation} \label{eqn-Y-W-intro}
\sum_{H\in \mcb{H}_{-j^*d^*}}\sum_{M \in \mcb{M}_{-j^*}}\sum_{i \in
[k_{j^*}]} \tilde{c}_{H,M,i}^2 \geq \frac{1}{2k_{j^*}} \left(
\sum_{H \in \mcb{H}_{-j^*d^*}}\sum_{M \in \mcb{M}_{-j^*}}
c^2_{H,MW_{1j^*}} \right),
\end{equation}
neglecting any contribution to the LHS of the above from $Q_0$ by our
a small enough choice of $\eta \ll \rho$ along with
\eqref{eqn:relative} and \eqref{eqn:relative3}.
The loss of $k_{j^*}$ factor in \eqref{eqn-Y-W-intro} is
compensated by the dependence of $\rho$ on $k$ as we shall see later. 
Combining
\eqref{eqn-Y-W-intro} with
\eqref{eqn-lemma-apply-intro-1}-\eqref{eqn-lemma-apply-intro-2} yields
\begin{equation}\label{eqn-non-noise-bd-intro}
\sum_{H\in \mcb{H}_{-j^*d^*}}\sum_{M \in \mcb{M}_{-j^*}}\sum_{i \in
[k_{j^*}]} \tilde{c}_{H,M,i}^2 \geq \Omega\left(1/k_{j^*}\right)
\rho^{d^*} \|Q_1\|_2^2.
\end{equation}
Consider now the sum 
$$\sum_{H\in \mcb{H}_{-j^*d^*}}\sum_{M \in
\mcb{M}_{-j^*}}\sum_{k_{j^*} < i \leq k} \tilde{c}_{H,M,i}^2.$$ 
Contribution
to the above can be from $Q_0$ or from $U_1Q_1^{(d^*+1)}$ -- the latter
due to the presence of $Y_{ij^*}$ ($k_{j^*} < i \leq k$) which
increases
the degree of $H \in \mcb{H}_{-j^*d^*}$ to $(d^* + 1)$ in the
representation of $Q_1$ over the basis $\mcb{H}\circ\mcb{M}$. 
Property (i) from our careful selection of $d^*$ is leveraged along with
our small enough choice of $\eta$ in \eqref{eqn:relative} along with
\eqref{eqn:relative3} to yield
\begin{equation}\label{eqn-noise-bd-intro}
\sum_{H\in \mcb{H}_{-j^*d^*}}\sum_{M \in
\mcb{M}_{-j^*}}\sum_{k_{j^*} < i \leq k} \tilde{c}_{H,M,i}^2 \leq 
O(1)
\rho^{d^*+1} \|Q_1\|_2^2.
\end{equation}
Using a choice $\rho \ll \eps/k$ we can combine the above with 
\eqref{eqn-non-noise-bd-intro} to obtain the following analog of
\eqref{eqn-1-bound-intro}:
\begin{equation}\label{eqn-second-1-bound-intro}
\sum_{H\in \mcb{H}_{-j^*d^*}}\sum_{M \in
\mcb{M}_{-j^*}}\sum_{i \in \mc{I}_{j^*}} \tilde{c}_{H,M,i}^2 \leq
\frac{\eps}{10} 
\sum_{H\in \mcb{H}_{-j^*d^*}}\sum_{M \in
\mcb{M}_{-j^*}}\sum_{i \in [k]} \tilde{c}_{H,M,i}^2,
\end{equation}
where $\mc{I}_{j^*} := \mc{I}\cap ([k]\times\{j^*\})$. 
Of course, since $\|Q_1\|_2 > 0$, we also obtain 
\begin{equation}\label{eqn-second-3-bound-intro}
\sum_{H\in \mcb{H}_{-j^*d^*}}\sum_{M \in
\mcb{M}_{-j^*}}\sum_{i \in [k]} \tilde{c}_{H,M,i}^2 > 0.
\end{equation}
The analysis above shows that for every good
choice of $\mc{I}$ there exist $(d^*, j^*)$ satisfying
\eqref{eqn-second-1-bound-intro}-\eqref{eqn-second-3-bound-intro}.
What remains is a probabilistic
concentration argument. 
Since $\Pr\left[\mc{I} \tn{ is good}\right] \geq \xi/2$, by averaging
we get that there exist $(d^*,j^*)$ and a fixing of
$\mc{I}\setminus\mc{I}_{j^*}$ such that with probability at least
$\xi/4Td$ over
the choice $\mc{I}_{j^*}$,
\eqref{eqn-second-1-bound-intro}-\eqref{eqn-second-3-bound-intro}
hold. Since each $i$ is added to $\mc{I}_{j^*}$ independently with
probability $\eps$, an application of Chernoff-Hoeffding shows that
the large deviation observed in 
\eqref{eqn-second-1-bound-intro} cannot occur with probability
$\xi/4Td$ (which is significant) unless the squared mass on the LHS of 
\eqref{eqn-second-3-bound-intro} is concentrated on a small number of
$i \in [k]$. This yields the desired decoding completing our sketch of
the analysis. The formal proof appearing in this work -- while
following the approach given above --  employs additional
notation and definitions for handling a few technicalities and ease of
presentation.

\medskip
\noindent
{\bf Combining $\mathfrak{P}_1$ with Label Cover and Folding.} The
test $\mathfrak{P}_1$ is executed on the $T$ blocks of coordinates
corresponding to $T$ randomly chosen vertices of a Smooth Label
Cover instance (as used in \cite{GRSW}). The resulting instance is then folded, i.e. 
the distribution on
the point-label pairs is projected onto a subspace $\mc{F}$ orthogonal to the
span of all the linear constraints implied by the edges of the Label
Cover. These linear constraints ensure that any vector in $\mc{F}$ has
equal mass sum in the coordinates of the two pre-images of a label 
given by an edge's projections. This property can be extended to polynomials
$P$ residing in $\mc{F}$. This fits with our decoding of
$\mathfrak{P}_1$ which is via a sub-polynomial
$\tilde{Q}_{\tn{lin}}$ linear in the variables
$\{Y_{ij^*}\}_{i=1}^k$ of the $j^*$th block. More specifically, we may
fix the vertices corresponding to all the blocks except the $j^*$th
and also the restriction of $\mc{I}$ to all the blocks except the
$j^*$th. This fixes $\mcb{H}_{-j^*d^*}\circ M_{-j^*}$ used
in \eqref{eqn-second-1-bound-intro}-\eqref{eqn-second-3-bound-intro}.
For a vertex $v$ let
$\tilde{c}_{H,M,i,v} = \tilde{c}_{H,M,i}$ when $v$ is chosen as the
$j^*$th vertex. Suppose for an edge between $u$ and $v$ (not among
the fixed vertices) the respective
pre-images of a common label are $A$ and $B$. Then, the folding
constraints imply
\begin{equation}
\sum_{i \in A}\tilde{c}_{H,M,i,u} = \sum_{i \in
B}\tilde{c}_{H,M,i,v}.\label{eqn-folding-intro}
\end{equation}
We combine the above with the decoding obtained from the analysis of
$\mathfrak{P}_1$ using appropriately set smoothness parameters 
to prevent masses in the pre-images containing
the decoded coordinates from cancelling out. 
The constraints  \eqref{eqn-folding-intro} then imply that the
decoded labels
define a labeling satisfying a significant fraction of edges of the Label Cover instance.
\vspace{-1em}

\paragraph{Organization.}
Section \ref{sec:prelim} presents some preliminaries. Section \ref{sec:reduction} describes
the reduction from Label Cover in the form of a PCP test. Section \ref{sec-folding} gives the 
constraints implied by folding extended to polynomials. 
 In Section \ref{sec-NO-case}, we show the soundness
of the reduction assuming a lemma (essentially restating \eqref{eqn-second-1-bound-intro}-\eqref{eqn-second-3-bound-intro}) 
about the structure of polynomials passing the test. The rest of the paper is
devoted to proving this lemma. In Section \ref{sec:anticonc}, we apply Gaussian anti-concentration to prove the analog 
of \eqref{eqn:relative}.
In Section \ref{sec:findingj}, we prove the structural lemma using Lemma \ref{lem-approx-structural-intro} as a key ingredient. Lemma 
\ref{lem-approx-structural-intro} is proved in Section \ref{sec:struct}.

\section{Preliminaries}\label{sec:prelim}

\subsection{The \LC Problem}

\begin{definition}[Smooth Label Cover]\label{sec:ug}
A \LC instance $\mc{L}(G(V, E),k,L, \{\pi_{e,v}\}_{e \in E,v \in e})$ consists of a regular connected graph with vertex set $V$ and edge set $E$, along with  projection maps $\pi_{e,v}: [k] \to [L]$ for all $e \in E,v \in e$. 
The goal is to find an assignment $\sigma: V \to [k]$ such that $\forall e = (u, w) \in E$, $\pi_{e,u}(\sigma(u)) = \pi_{e,w}(\sigma(w))$. The optimum for a \LC instance is the maximum fraction of edges satisfied by an assignment.
\end{definition}

The following Theorem from \cite{GRSW} states the hardness of \LC problem:

\begin{theorem}\label{thm-LC-hardness}
There exists a constant $c_0 > 0$ such that for any constant integer
parameters $J,R \ge 1$, it is \NP-hard to distinguish between the
following cases for a \LC instance $\mc{L}$$(G(V, E),$ $k,$ $L,$
$\{\pi_{e,v}\}_{e \in E,v \in e})$ with parameters $k = 7^{(J+1)R}, L=2^{R}7^{JR}$.
\begin{itemize}
	\item {\bf YES}: There is a labeling that satisfies every edge.
	\item {\bf NO}: Every labeling satisfies less than $2^{-c_0R}$-fraction of edges.
\end{itemize}
Additionally, the instance $\mc{L}$ satisfies the following properties:
\begin{itemize}
	\item {\bf Smoothness}: For any $v \in V$, and labels $i,j \in [k], i \ne j$, $\Pr_{e \sim v}[\pi_{e,v}(i) = \pi_{e,v}(j)] \le 1/J$. In particular, for a subset $S \subseteq [k]$, $\Pr_{e \sim v}\left[\left|\pi_{e,v}(S)\right| = \left|S\right|\right] \leq |S|^2/(2J)$. 
	\item The degree $d_{\mc{L}}$ of the graph $G$ is a constant dependent only on $J$ and $R$.
	\item For any vertex $v \in V$, edge $e \in E$ incident on
vertex $v$, and $j \in [L]$, we have
$\big|\big(\pi_{e,v}\big)^{-1}(j)\big| \le t_{\mc{L}} := 4^R$.
	\item {\bf Weak Expansion}: For any $V^\prime \subseteq V$, the number of edges induced in $V^\prime$ is at least $\frac{\delta^2}{2}|E|$ where $\delta = |V^\prime|/|V|$.
\end{itemize}
\end{theorem}

\subsection{Hermite Bases for Multivariate Polynomials}
For integer $d \geq 0$, the {\em Hermite polynomials} $H_d(x)$ are degree-$d$ univariate polynomials such that \\ $\E_{X \sim N(0,1)}[H_d(X)^2]=1$ and $\E_{X \sim N(0,1)}[H_d(X) H_{d'}(X)]=0$ for any $d \neq d'$. For example, $H_0(x) = 1$, $H_1(x) = x$, $H_2(x) = \frac{1}{\sqrt{2}}(x^2 - 1)$, and $H_3(x) = \frac{1}{\sqrt{6}}(x^3 - x)$. 

For $\mbf{d} \in \N^n$, we define $H_{\mbf{d}}(x_1, \dots, x_n) = \prod_{i \in [n]} H_{d_i}(x_i)$. For $D \geq 0$, let $\mc{H}_D = \{H_{\mbf{d}} : \mbf{d} \in \N^n, \sum_{i \in [n]} d_i \leq D\}$ denote the {\em Hermite basis for degree-$D$ polynomials}. The following is immediate.

\begin{fact}
The set $\mc{H}_D$ forms an orthonormal basis for $n$-variate degree-$D$ polynomials whose inputs are drawn from $N(0,1)^n$. In particular, for any $P: \R^n \to \R$ of degree $\leq D$, we can write:
$$P(x) = \sum_{\mbf{d} \in \N^n: \sum_i d_i \leq D} \hat{f}(\mbf{d}) \cdot H_{\mbf{d}}(x)$$ 
and moreover, $\E_{x} P(x) = \hat{f}(\mbf{0})$ and $\E_x[P(x)^2] = \sum_{\mbf{d}} \hat{f}^2(\mbf{d})$.
\end{fact} 
\subsection{Concentration and Anti-Concentration}\label{sec-conc-anticonc}
The magnitude of polynomials in our analysis is controlled using the following standard bound.

\medskip
\noindent
\textbf{Chebyshev's Inequality.} For any random variable $X$ and $t > 0$, $\Pr\left[|X| > t\right] \leq \slfrac{\E[X^2]}{t^2}$.

\medskip
\noindent
The above is used in conjunction with Carbery and Wright's~\cite{CW01} powerful anti-concentration bound for polynomials over independent Gaussian variables.
\begin{theorem}	\label{thm:carbery-wright-prelim}
\tn{(Carbery-Wright~\cite{CW01})} Suppose $P: \R^\ell \to \R$ is a degree-$d$ polynomial over independent $N(0,1)$ random variables. Then,
$$\Pr\left[|P| \leq \eps \|P\|_2\right] = O(d\eps^{1/d}).$$
\end{theorem}
In addition, we also use following Chernoff-Hoeffding bound.
\begin{theorem}[Chernoff-Hoeffding]\label{thm:chernoff}
 Let $X_1,\dots, X_n$ be independent random variables, each bounded as $a_i \leq X_i \leq b_i$ with $\Delta_i = b_i - a_i$ for $i = 1,\dots, n$. Then, for any $t > 0$,
	$$\Pr\left[\left|\sum_{i=1}^n X_i - \sum_{i=1}^n\E[X_i]\right| > t\right] \leq 2\cdot\tn{exp}\left(-\frac{2t^2}{\sum_{i=1}^n\Delta_i^2}\right).$$
\end{theorem}

\section{Hardness Reduction}\label{sec:reduction}
\begin{figure}[t]
\begin{mdframed}
\begin{center}\textbf{The Basic PCP Test given instance $\mc{L}$ of \LC}\end{center}

\begin{enumerate*}

\item For each $j \in [T]$, the test chooses $T$ random vertices $v_1,v_2,\ldots,v_T \overset{u.a.r.}{\sim} V$.
Let $Y_{ij} := Y^{v_j}_i$.

\item Sample $\{\delta_{j}\, \mid\, j\in[T]\}$ from the joint Gaussian
	distribution where the marginals are $N(0,1)$,
	$\E[\delta_j\delta_{j'}] = -1/(T-1)$ for all $j\neq j'$, 
	and $\sum_{j=1}^T \delta_{j} = 0$.

\item Sample $b \in \{-1,1\}$ uniformly at random.

\item Sample $\mc{I} \subseteq [k]\times[T]$ to be a random subset
where each $(i,j) \in [k]\times[T]$ is added to $\mc{I}$
independently with probability $\eps$.

\item For each $(i,j) \in ([k]\times[T])\setminus \mc{I}$, set $Y_{ij} := \sqrt{{(T-1)}/{T}}\cdot \delta_{j} +
b\eta.$

\item Independently for each $(i,j) \in \mc{I}$,
sample  $Y_{ij}$ from $N(0,{1})$. 

\item Set the variables of all other vertices (except $\{v_j
\mid\, j \in [T]\}$) to be $0$. Let this setting of the variables be the point ${\bf y} \in \mathbb{R}^{\mc{Y}}$.

\item Output the point-sign pair $({\bf y},b)$.

\end{enumerate*}
\end{mdframed}
\caption{Basic PCP Test}
\label{fig:pcptest}
\end{figure}

The following reduction from \LC  directly implies our main theorem.
\begin{theorem}\label{thm-hardness-redn}
For any $\xi > 0$ and $d \in \mathbb{Z}^+$, there exists a choice of $R$ and $J$ in Theorem \ref{thm-LC-hardness} and
a  polynomial-time reduction from the corresponding \LC instance $\mc{L}$ to a set of point-sign pairs $\mc{Q} \subseteq \mathbb{R}^N\times \{-1,1\}$ such that:
\begin{itemize*}
\item \textbf{YES Case.} If $\mc{L}$ is a YES instance, then there exists
a linear form $L$ satisfying
$$\Pr_{({\bf x}, s) \in \mc{Q}} \left[\tn{sign}\left(L({\bf x})\right) = s\right]
\geq 1 - \xi.$$

\item \textbf{NO Case.} If $\mc{L}$ is a NO instance, then for any degree-$d$
polynomial $P$
$$\Pr_{({\bf x}, s) \in \mc{Q}} \left[\tn{sign}\left(P({\bf x})\right) = s\right]
\leq \frac{1}{2} + \xi.$$
\end{itemize*}

\end{theorem}

The last sentence of \cref{thm:main} is justified in Section \ref{sec:ends}.

\subsection{The Basic PCP Test}\label{sec-basic-pcp-test}

We begin with a Basic PCP Test given an instance
 $\mc{L}(G(V, E),k,L, \{\pi_{e,v}\}_{e \in E,v \in e})$ of \LC.  
For each vertex $v
\in V$, there is a set of variables $\{Y^v_i\}_{i=1}^k$, and the set
of all the variables $\mc{Y}$ is a union over all vertices $v \in V$ of these variable sets. 
The test is described by  
the sampling procedure in Figure \ref{fig:pcptest}, and yields a distribution over point-sign pairs which is independent of the constraints in $\mc{L}$. It uses some additional parameters set as follows: $T \coloneqq 10d$, $\eps \coloneqq (\xi/32Td)$, $\eta \coloneqq \Big(\frac{\eps\xi}{20kdT}\Big)^{d6^{3d}}$, where $d$ is from the statement of Theorem \ref{thm-hardness-redn}.

\subsubsection{Folding over constraints of $\mc{L}$} \label{sec-folding}
To ensure consistency across the edges of $\mc{L}$, the points generated by the Basic PCP Test are \emph{folded} over a specific subspace.  The points generated by the Basic PCP Test reside in the space $\mathbb{R}^{\mc{Y}}$. 
Now, for a fixed $e=(u,w) \in E$ and $j \in [L]$, we define the vector ${\bf h}^{e}_j \in \mathbb{R}^{\mc{Y}}$ as 
\begin{align}
	{\bf h}^e_j(Y^v_i) &= \begin{cases}
		1 & \mbox{ if } v = u\tn{ and } i \in \left(\pi_{e,u}\right)^{-1}(j), \\
		-1 &\mbox{ if } v = w\tn{ and } i \in \left(\pi_{e,w}\right)^{-1}(j), \\
		0 & \mbox{otherwise.}
		\end{cases}
\end{align}
Let $\mc{H} \subseteq \mathbb{R}^{\mc{Y}}$ be the subspace formed by the linear span of the vectors $\{{\bf h}^e_{j}\}_{e \in E, j \in [L]}$, and let $\mc{F}$ be the orthogonal complement of $\mc{H}$ in $\mathbb{R}^{\mc{Y}}$, i.e. $\mathbb{R}^{\mc{Y}} = \mc{H}\oplus \mc{F}$ and $\mc{H}\perp \mc{F}$. For each point-sign pair $({\bf y},b)$ generated by the Basic PCP Test, construct $(\ol{{\bf y}},b)$ where $\ol{\bf y}$ is the projection of ${\bf y}$ onto the subspace $\mc{F}$, represented in some (fixed) orthogonal basis for $\mc{F}$. 

Conversely, for any vector $\ol{\bf z} \in \mc{F}$, let ${\bf z}$ be its representation in $\mathbb{R}^{\mc{Y}}$. It is easy to see that such a ${\bf z}$ satisfies: for every $e=(u,w) \in E$ and  $j \in [L]$, $\langle {\bf z}, {\bf h}^{e}_j \rangle = 0$ which is equivalent to
\begin{equation}					\label{eq:folding_constraint}
\tn{\it Constraint }\ \ {\mc{C}_{e,j}:} \qquad\qquad \sum_{i \in \left(\pi_{e,u}\right)^{-1}(j)} {\bf z}(Y^u_i)  = \sum_{i \in \left(\pi_{e,w}\right)^{-1}(j)} {\bf z}(Y^w_i). 
\end{equation}
For our purpose we shall extend the above constraint to polynomials as well. Consider a polynomial $Q$ in $\mathbb{R}^{\mc{Y}}$. For any monomial $M$ over the variables $\mc{Y}$, let $c_{Q,M}$ be its coefficient in $Q$. Fix an edge $e = (u,w)$ and $j \in [L]$, and a monomial $M$ such that $M$ does not contain any variable from the set $\{Y^u_i\, \mid\, i\in \left(\pi_{e,u}\right)^{-1}(j)\}\cup \{Y^w_i\, \mid\, i\in \left(\pi_{e,w}\right)^{-1}(j)\}$. For such a choice of $e, j,$ and $M$ we say that $\mc{C}_{e,j,M}$ is a \emph{valid} constraint where:
\begin{equation}					\label{eq:folding_constraint_poly}
\tn{\it Constraint }\ \ {\mc{C}_{e,j,M}:} \qquad\qquad \sum_{i \in \left(\pi_{e,u}\right)^{-1}(j)} c_{Q,M\cdot Y^u_i}  = \sum_{i \in \left(\pi_{e,w}\right)^{-1}(j)} c_{Q,M\cdot Y^w_i}. 
\end{equation}
We have the following lemma.
\begin{lemma}\label{lem-folding-poly}
Let $\ol{Q}$ be a polynomial that resides in $\mc{F}$, i.e. is represented in an orthogonal basis\footnote{A polynomial $\ol{Q}$ being represented in an orthogonal basis for a subspace $\mc{F}$ means $\ol{Q}$ can be written as a polynomial over the linear forms corresponding to an orthogonal basis for $\mc{F}$.} for $\mc{F}$, and let $Q$ be its representation in $\mathbb{R}^{\mc{Y}}$. Then, $Q$ satisfies all valid constraints $\mc{C}_{e,j,M}$.
\end{lemma}
\begin{proof}
Suppose for a contradiction $Q$ does not satisfy a valid constraint $\mc{C}_{e, j, M}$. Consider the vector ${\bf r}$ where,
$${\bf r}(Y^v_i) = \begin{cases}
c_{Q,M\cdot Y^u_i} & \tn{ if } v=u, i \in \left(\pi_{e,u}\right)^{-1}(j)\\
c_{Q,M\cdot Y^w_i} & \tn{ if } v=w, i \in \left(\pi_{e,w}\right)^{-1}(j)\\
0 & \tn{ otherwise.}
\end{cases}$$
Since Equation \eqref{eq:folding_constraint_poly} is not satisfied, it is easy to see that $\langle {\bf r}, {\bf h}^e_j\rangle \neq 0$, and thus ${\bf r} = {\bf r}_0 + {\bf r}_1$ where  ${\bf r}_0 \in \mc{F}$ and ${\bf r}_1 \in \mc{H}$. On the other hand, consider an orthogonal basis $\mc{B}$ for $\mathbb{R}^{\mc{Y}}$ that is an extension of $\{{\bf r}_1\}$, i.e. ${\bf r}_1$ is an element of $\mc{B}$. $P$ can now be represented as:
$$P \equiv  {\bf r}_1[\mc{Y}]\cdot P_1 + P_0,$$
where $P_1$ is a polynomial represented in $\mc{B}$, $P_0$ is represented in $\mc{B}\setminus\{{\bf r}_1\}$, and ${\bf r}_1[\mc{Y}]$ is the $\mc{Y}$-linear form $\sum_{Y \in \mc{Y}} {\bf r}_1(Y)\cdot Y$. Note that $P_1$ is not identically zero, in particular it contains the monomial $M$.
This implies that $P$ cannot be represented over any basis for $\mc{F}$, which is a contradiction.  
\end{proof}
\begin{remark}\label{rem-folding-poly}
Instead of monomials $M$, the constraints in \eqref{eq:folding_constraint_poly} analogously hold for elements $B$ of a basis $\mcb{B}$ for polynomials over any set of variables not containing $\{Y^u_i\, \mid\, i\in \left(\pi_{e,u}\right)^{-1}(j)\}\cup \{Y^w_i\, \mid\, i\in \left(\pi_{e,w}\right)^{-1}(j)\}$.
\end{remark}

\subsection{ The Final PCP Test}\label{sec-final-pcp-test}
Given a degree-$d$ polynomial $\ol{P}_{\tn{global}}$ over the space $\mc{F}$, the test samples $({\bf y},b)$ from the Basic PCP Test (as described in Figure \ref{fig:pcptest}), and constructs $(\ol{\bf y},b)$ as described in Section \ref{sec-folding}. The test accepts {\it iff} ${\rm sign}\left(\ol{P}_{\tn{global}}(\ol{\bf y})\right)=b$.

\begin{remark}The Basic PCP Test generates a distribution over $\mathbb{R}^{\mc{Y}}\times \{-1,1\}$ using
various independently Gaussian random variables. Therefore, the support set of this distribution is not finite. In Section \ref{sec:ends}, using techniques from \cite{DOSW11}, we discretize the Basic PCP Test. 
Building upon the discretized Basic PCP Test, the Final PCP Test yields the desired finite subset $\mc{Q}$ in polynomial time. 
\end{remark}

\subsection{Completeness Analysis}
Suppose there is a labeling $\sigma :V \to [k]$ which satisfies all the edges of $\mc{L}$. Define $L^*(\mc{Y})
= \sum_{v \in V} Y^v_{\sigma(v)}$ to be a linear form. Note that $L^*({\bf y}) := \langle {\bf r}^*, {\bf y}\rangle$ for some ${\bf r}^* \in \mc{F}$, and so $L^*$ can be represented in an orthogonal basis for $\mc{F}$. Thus, for any point ${\bf y} \in \mathbb{R}^{\mc{Y}}$, $L^*(y) = L^*(\ol{\bf y})$ where $\ol{\bf y}$ is the projection of ${\bf y}$ on to $\mc{F}$ as defined in Section \ref{sec-folding}.

Now consider $({\bf y}, b)$ generated by the Basic PCP Test.
By a union bound over the randomness of the test,  with probability at least $(1 - \eps T)$:  $(\sigma(v_j), j) \not\in \mc{I}$ for each $j \in [T]$. Given this, it is easy to see that $L^*({\bf y}) = b$, and by the above reasoning $L^*(\ol{\bf y}) = b$.  
Thus, $L^*$ satisfies the Final PCP Test with probability at least $(1 - \eps T)$. Our choice of $\eps$ yields the desired accuracy.

\section{Soundness Analysis}\label{sec-NO-case}
Given the \LC instance $\mc{L}$, suppose that there is a degree-$d$ polynomial (over $\mc{F}$) $\ol{P}_{\tn{global}}$ such that
the Final PCP Test accepts with probability $1/2 + \xi$. Our goal in the rest of this paper is to show that in this case there exists a labeling that satisfies at least $2^{-c_0R}$-fraction of the edges of $\mc{L}$, for an appropriate choice of constants $R$ and $J$ in Theorem \ref{thm-LC-hardness} and because of its NO Case we would be done.

Let $P_{\tn{global}}$ be the representation of $\ol{P}_{\tn{global}}$  in $\mathbb{R}^{\mc{Y}}$, so that  $\ol{P}_{\tn{global}}(\ol{\bf y}) = P_{\tn{global}}({\bf y})$ where $\ol{\bf y} \in \mc{F}$ is a point generated by the Final PCP Test from a point ${\bf y}$ generated by the Basic PCP Test as given in Section \ref{sec-final-pcp-test}. Therefore, $P_{\tn{global}}({\bf y}) = b$ with probability at least $1/2 + \xi$ over the pairs $({\bf y}, b)$ output by the Basic PCP Test. Using this, we focus on analyzing the structure of $P_{\tn{global}}$.

To begin the analysis note that with probability at least $2\xi$ over
the choices of the verifier other than $b$, $P_{\tn{global}}$ flips
its sign on flipping $b$.  Call a choice of $\{v_j\,\mid\, j \in [T]\}$ 
\emph{good} if conditioned on this, the same holds with
probability at least $\xi$ over the rest of the choices (other than $b$) of the
verifier. 
By averaging, with probability at least $\xi$,
the verifier makes a good choice. We now fix such a 
good choice $\{v_j\,\mid\, j \in [T]\}$.

For convenience, we shall use $P$ to denote the restriction of
$P_{\tn{global}}$ to ${\bf Y} := \{Y_{ij}\, \mid\, i\in[k], j\in [T]\}$. 
Let $\mc{D}$ be the distribution on $({\bf Y}, b)$ generated by the
steps of the verifier.
 Our analysis shall first show that  
in terms of this basis $P$ must have a  certain structure which  
will then be used to determine a good labeling for $\mc{L}$.

\subsection{Basis Transformations}
For the purpose of the analysis, we shall rewrite the variables
${\bf Y}$ in different bases. Before we
do that, we shall isolate the noisy set $\mc{I}$ of the Basic PCP Test.

\subsubsection{Choice of set $\mc{I}$}
The distribution $\mc{D}$ involves choosing the set $\mc{I}$ in which
each $(i,j)$ is added independently at random with probability $\eps$.
Let us call a setting of $\mc{I}$ as \emph{nice} if it satisfies:
\begin{enumerate}
\item For each $j$, $\left|\{i\,\mid\, (i,j)\in\mc{I}\}\right| \leq
k/2$.
\item 
With probability $\xi/2$ over the rest of the choices of the verifier
(except $b$), $P$ flips its sign on flipping $b$. 
\end{enumerate}
By our setting of $\eps$ and $T$, for a large enough value of $k$, and
applying the Chernoff Bound, a union bound and an averaging argument,
we have:
\begin{equation}
\Pr_{\mc{D}}\left[\mc{I} \tn{ is nice}\right] \geq \xi/4.
\label{eqn-nice-I}
\end{equation}

Going forward, we shall fix a nice choice of $\mc{I}$. By relabeling, we
may assume that there exist $k/2 \leq k_j \leq k$ for $j \in [T]$ such that
\begin{equation}
\mc{I} = \bigcup_{j=1}^T \{(i,j) \,\mid\, i = k_j+1,\dots, k\}.
\end{equation}
Based on this  {nice} choice of $\mc{I}$, we now define new bases
for the ${\bf Y}$ variables. Let $\mc{D}_{\mc{I}}$
denote the distribution of the variables after fixing a nice $\mc{I}$.

\subsubsection{Bases {\bf W} and {\bf U}} \label{sec:R=1Basis}
For each $j \in [T]$, we define $(W_{1j}, W_{2j}, \dots, W_{k_jj})$ as a fixed orthogonal transformation of $(Y_{1j}, Y_{2j}, \dots, Y_{k_jj})$ so that
\begin{equation}
W_{1j} = \frac{1}{ k_j} \sum_{i=1}^{k_j} Y_{ij},  \ \ \ \ \tn{and} \ \ \ \ 
W_{ij} = \sum_{\ell \in [k_j]}c_{i\ell}Y_{\ell j}\ \tn{for all $i \in [2, k_j]$}
\end{equation} 
where the vectors $\left\{{\bf c}_i =[c_{i1}, c_{i2}, \ldots ,c_{ik_j}]^{\sf{T}}\right\}_{i=2}^{k_j}$ satisfy 
\begin{itemize*}
	\item For all $i,i^\prime \in [k_j] \setminus \{1\}$ we have $\langle {\bf c}_i, {\bf c}_{i^\prime} \rangle = 0$
	\item Each vector ${\bf c}_i$ satisfies $\|{\bf c}_i\|^2 = 1/k_j$ and ${\bf c}_i \perp \mathbbm{1}$ where $\mathbbm{1}$ is the all ones vector in $\mathbb{R}^{k_j}$.
\end{itemize*}
We shall also define the vector ${\bf c}_1 = \frac{1}{k_j}\cdot\mathbbm{1}$ where $\mathbbm{1} \in \R^{k_j}$ is the the vector of all ones.
The above along with the distribution of $\{Y_{ij}\,\mid\, i = 1,\dots,k_j\}_{j=1}^T$ in $\mc{D}_{\mc{I}}$ directly implies the following.
\begin{lemma}\label{lem:W-dist}
Under the distribution $\mc{D}_\mc{I}$:
\begin{enumerate*}
\item[(i)]
$W_{1j} = \sqrt{\nfrac{(T-1)}{T}} \cdot \delta_j + b\eta$ 
\item[(ii)]  For $i \neq 1$, $W_{ij} = 0$.
\end{enumerate*}
\end{lemma}
Let $U_1, \dots, U_T$ be a fixed orthonormal transformation of $(W_{11}, \dots, W_{1T})$, where
\begin{equation}
U_1 = \frac{1}{\sqrt{T}} \sum_{j=1}^T W_{1j}, \ \ \ \ \tn{and}\ 
U_t = \sum_{j \in [T]}a_{tj}W_{1j} \ \tn{for all $t \in [2, T]$}
\end{equation}
where vectors ${\bf a}_2,\ldots,{\bf a}_T$ are orthonormal and each vector ${\bf a}_t = [a_{t1}, a_{t2}, \ldots ,a_{tT}]^{\sf{T}}$  satisfies   $\sum_{j \in T}{a}_{tj} =0$ (i.e., they are orthogonal to the all ones vector).
\begin{lemma}\label{lem:U-dist}
Under the distribution $\mc{D}_{\mc{I}}$,
\begin{itemize*}
\item[(i)] $U_1 = {b\eta}{\sqrt{T}}$
\item[(ii)] For each $1 < t \leq T$, $U_t \sim N(0,1)$ i.i.d.  
\end{itemize*}
\end{lemma}
\begin{proof}
Lemma \ref{lem:W-dist} along with the definition of $U_1$ yields the first part. 
The second part follows from an application of  Lemma \ref{lem-ortho-transform}.
\end{proof}
Before we proceed, we briefly summarize the variables and their
distribution under $\mc{D}_{\mc{I}}$.

\begin{itemize}
	\item {\bf Noisy Indices} For a fixed $j \in [T]$, $[k_j]$ is the set of non-noisy $i$'s where $k_j \geq k/2$.
	\item {\bf The $Y$-variables} . For each $(i,j) \in [k] \times [T] \setminus\mc{I}$, $Y_{ij} =  \sqrt{\nfrac{(T-1)}{T}}\cdot\delta_j +b\eta$. For $(i,j) \in \mc{I}$, $Y_{ij}$'s are independent $N(0,1)$ random variables. 
	\item {\bf The $W$-variables} For a fixed $j$, we define variables $W_{1j},\ldots,W_{k_jj}$ with $W_{1j} =  \sqrt{\nfrac{(T-1)}{T}}\cdot\delta_j + b\eta$ and $W_{2j},\ldots,W_{{k_j}j}$ are $0$. 
	\item {\bf $U$-variables} We define $U_1 = \frac{1}{\sqrt{T}}\sum_{j \in [T]}W_{1j}$ which is $b\eta\sqrt{T}$ and is independent of the 
variables $U_2,\ldots,U_T$ where each $U_t$ is i.i.d. $N(0,1)$ for $t>1$.
\end{itemize}

\subsection{A Hybrid Basis Relative to $j^*$ and $d^*$}\label{sec:hybrid}
Recall that we have fixed a nice $\mc{I}$. In this section, we define a basis for polynomials using a fixed choice of $j^* \in [T]$ and $d^* \in [d]$.  For convenience let $[T_{-j^*}] := [T]\setminus\{j^*\}$.

\begin{definition}\label{def-hermite-basis-partial}
Let $\mcb{H}_{-j^*}$ be the Hermite basis for
all polynomials over the independent Gaussian variables $\{Y_{ij}\,\mid\, i \in [k]\setminus[k_j], j\in[T_{-j^*}]\}$.
In particular, $\E[H^2] = \E[G^2] = 1$ and $\E[HG] = 0$ for each
$H, G \in \mcb{H}_{-j^*}$, $H\neq G$. Let $\mcb{H}_{-j^*d^*}$  be the
set of basis elements of $\mcb{H}_{-j^*}$ of degree exactly $d^*$. 
\end{definition}
\begin{definition}\label{def-monomial-Ws-partial}
Let $\mcb{M}_{-j^*}$ be the \emph{standard monomial basis} for
polynomials over the variables $\{W_{1j}\,\mid\,j\in [T_{-j^*}]\}$. In
particular, each element of $\mcb{M}_{-j^*}$ is of the form
$\prod_{j\in [T_{-j^*}]}W_{1j}^{a_j}$ for some non-negative integers
$a_j$ ($j \in [T_{-j^*}]$).
\end{definition}
\begin{definition}\label{def-combined-basis}
Let $\mcb{B}_{-j^*} := \mcb{H}_{-j^*} \circ \mcb{M}_{-j^*}$ be the
combined basis for polynomials over the variables of 
$\mcb{H}_{-j^*}$ and $\mcb{M}_{-j^*}$, where each element $B$ is of
the form $HM$ for some $H \in \mcb{H}_{-j^*}$ and $M \in
\mcb{M}_{-j^*}$ and $\tn{deg}(B) = \tn{deg}(H) + \tn{deg}(M)$. For
convenience  we also define the subset $\mcb{B}_{-j^*d^*} := 
\mcb{H}_{-j^*d^*} \circ \mcb{M}_{-j^*}$, i.e. each element of 
$\mcb{B}_{-j^*d^*}$ is of the form $HM$ where $H \in
\mcb{H}_{-j^*d^*}$ and $M \in \mcb{M}_{-j^*}$.
\end{definition}
Lastly, let $\mcb{S}_{j^*}$ be the set of all multisets of
$R_{j^*} = \{(i,j^*)\,\mid\, i\in [k]\}$. For an element $S \in
\mcb{S}_{j^*}$, let $S(i,j^*)$ denote the number of occurrences of
$(i,j^*)$ in $S$. Using this, we define $Y_S := \prod_{(i,j^*) \in R_{j^*}} Y_{ij^*}^{S(i,j^*)}$.

\medskip
Writing the polynomial $P$  in the basis given by products of $\mcb{B}_{-j^*}$, $\{W_{ij}: j \in [T_{-j^*}], i \in [k_j]\setminus\{1\}\}$ and $\{Y_{ij^*}: i \in [k]\}$, 
the polynomial $P$ can be  
represented as:
\begin{equation}
P = P_{\tn{omit}} + \sum_{\substack{S\in \mcb{S}_{j^*} \\ B\in \mcb{B}_{-j^*}}}
c_{S, B}  Y_S B, \label{eqn-P-hybrid-rep}
\end{equation}
where $c_{S,B}$ are constants and\footnote{The reason for treating $P_{\tn{omit}}$ separately is that it vanishes under the distribution $\mc{D}_{\mc{I}}$.} $P_{\tn{omit}}$ is the sub-polynomial of $P$ consisting of all monomials containing a variable from $\{W_{ij}: j \in [T_{-j^*}], i \in [k_j]\setminus\{1\}\}$. Of course, since $P$ is of degree at
most $d$, the only terms that occur in the above sum satisfy
$\tn{deg}(B) + |S| \leq d$.

For a fixed $0 \leq d^* \leq d-1$ we will be interested in capturing the
the mass of $P$ linear in $Y_{ij^*}$ and the subset
$\mcb{B}_{-j^*d^*}$. Abusing notation
to let $c_{(i,j^*), B} = c_{(S, B)}$ where $S = \{(i,j^*)\}$ is the
singleton multiset, define
\begin{equation}
c_{i,j^*,d^*} = \sqrt{\sum_{B \in \mcb{B}_{-j^*d^*}} c_{(i,j^*), B}^2} 
\label{eqn-lin-coeff}
\end{equation} 
for each $(i,j^*) \in R_{j^*}$ and $0 \leq d^* \leq d-1$. 

\subsection{Main Structural Lemma}
We are now ready to describe the structure that $P$ must exhibit in order to pass the Basic PCP test. 
Let us first define a {\em distinguished pair} $(j^*, d^*)$ for a fixed setting of $\mc{I}$.

\begin{definition}\label{def-distinguished-jd}
A pair $(j^*, d^*) \in [T]\times\{0,\dots,d-1\}$ is said to be
\emph{distinguished} for $\mc{I}$ if,
\begin{equation}
\sum_{(i,j^*) \in \mc{I}} c_{i,j^*,d^*}^2 \leq
\frac{\eps^4}{4} \cdot \left(\sum_{(i,j^*) \in
([k]\times\{j^*\})\setminus\mc{I}} c_{i,j^*,d^*}^2\right), \label{eqn-main-nice-I-1}
\end{equation}
and,
\begin{equation}
\sum_{(i,j^*) \in
([k]\times\{j^*\})\setminus\mc{I}} c_{i,j^*,d^*}^2 > 0. \label{eqn-main-nice-I-3}
\end{equation}
Here, $\eps$ is the noise parameter used in the PCP test.
\end{definition}

The main lemma that we prove is the following.

\begin{lemma}[Main Structural Lemma]\label{lem-main-nice-I}
For every \emph{nice} choice of $\mc{I}$, there exists $j^* \in [T]$
and $d^* \in \{0,1,\dots, d-1\}$ such that $(j^*, d^*)$ is
distinguished for $\mc{I}$.
\end{lemma}
The proof of the above lemma is given in Section \ref{sec:findingj}
building upon analysis in Section \ref{sec:anticonc}. Both Sections
\ref{sec:anticonc} and \ref{sec:findingj} assume a setting of nice $\mc{I}$.

Using \eqref{eqn-nice-I} and a simple averaging, 
the above lemma implies that there exists $(j^*, d^*)$ such that:
\begin{equation}
\Pr_{\mc{I}}\left[(j^*,d^*) \tn{ is distinguished for }\mc{I}\right] \geq
\frac{\xi}{4Td}.
\label{eqn-distinguished-I}
\end{equation}

\subsection{Implications of the Structural Lemma}
We now fix $(j^*, d^*)$ satisfying
\eqref{eqn-distinguished-I}. Let us consider the random choice of $\mc{I}$ as
first picking $\mc{I}_{-j^*} := \mc{I}\cap([k]\times ([T]\setminus
\{j^*\}))$, and then picking $\mc{I}_{j^*} := \mc{I}\cap([k]\times
\{j^*\})$. Note that the choice of $\mc{I}_{j^*}$ is independent of
$\mc{I}_{-j^*}$.
Call a choice of $\mc{I}_{-j^*}$ as \emph{shared-heavy} if,
\begin{equation}
\Pr_{\mc{I}_{j^*}}\left[(j^*,d^*) \tn{ is distinguished for
}\mc{I}_{j^*}\cup\mc{I}_{-j^*}\right] \geq \frac{\xi}{8Td}. 
\label{eqn-shared-heavy-def}
\end{equation}
From \eqref{eqn-distinguished-I} and an averaging argument we
have:
\begin{equation}
\Pr_{\mc{I}_{-j^*}}\left[\mc{I}_{-j^*} \tn{ is shared-heavy}\right]
\geq \frac{\xi}{8Td}.
\label{eqn-shared-heavy-prob}
\end{equation}

Let us fix a shared-heavy $\mc{I}_{-j^*}$. Note that with this fixing,
the bases given in Section \ref{sec:hybrid} are well defined, and in
particular $P$ can be represented as in
\eqref{eqn-P-hybrid-rep}. Since there is at least one choice of
$\mc{I}_{j^*}$ such that $(j^*, d^*)$ is distinguished for
$\mc{I}_{j^*}\cup \mc{I}_{-j^*}$, using \eqref{eqn-main-nice-I-3} this implies
\begin{equation}
\sum_{i \in
[k]} c_{i,j^*,d^*}^2 > 0. \label{eqn-main-shared-heavy-I-3}
\end{equation}
Further we have the following lemma. (This is where we are finally randomizing over $\mc{I}_{j^*}$.)
\begin{lemma}\label{lem:Azuma-apply}
There exists $i^* \in [k]$ such that,
$$c_{i^*,j^*,d^*}^2 \geq \nu^2 
\left(\sum_{i \in [k]} c_{i,j^*,d^*}^2\right),$$
for $\nu = \eps^2/2$. 
\end{lemma}
\begin{proof}
Assume that there is no such $i^*$ as in the lemma.
Over the choice of ${\mc{I}_{j^*}}$, consider the random variable 
$\sum_{(i,j^*) \in \mc{I}_{j^*}} c_{i,j^*,d^*}^2$. The
contribution from each $i$ to this sum is independently $0$ with
probability $(1- \eps)$ and
$c_{i,j^*,d^*}^2$ with probability $\eps$.
Thus,
$$\E_{\mc{I}_{j^*}}\left[\sum_{(i,j^*) \in \mc{I}_{j^*}}
c_{i,j^*,d^*}^2\right] = \eps\left(\sum_{i \in [k]}
c_{i,j^*,d^*}^2\right).$$
Now,
{\allowdisplaybreaks
\begin{align*}
&\Pr\left[\sum_{(i,j^*) \in \mc{I}}
c_{i,j^*,d^*}^2 \leq (\eps/2)\left(\sum_{(i,j^*) \in
[k]\times\{j^*\}\setminus \mc{I}}
c_{i,j^*,d^*}^2\right) \right] \nonumber \\
&\leq \Pr\left[\sum_{(i,j^*) \in \mc{I}_{j^*}}
c_{i,j^*,d^*}^2 \leq (\eps/2)\left(\sum_{i \in [k]}
c_{i,j^*,d^*}^2\right) \right] \nonumber \\
&\leq  
\Pr\left[\left|\sum_{(i,j^*) \in \mc{I}_{j^*}}
c_{i,j^*,d^*}^2 - \E\left[\sum_{(i,j^*) \in \mc{I}_{j^*}}
c_{i,j^*,d^*}^2\right]\right| \geq (\eps/2)\left(\sum_{i \in [k]}
c_{i,j^*,d^*}^2\right) \right]\\
&\overset{1}{\leq} 2\cdot\tn{exp}\left(-\frac{2(\eps/2)^2\cdot \left(\sum_{i \in [k]}
c_{i,j^*,d^*}^2\right)^2}{\sum_{i\in[k]} c_{i,j^*,d^*}^4}\right)
\nonumber \\
&\leq  2\cdot \tn{exp}\left(-\frac{(\eps^2/2) \cdot \left(\sum_{i \in [k]}
c_{i,j^*,d^*}^2\right)^2}{\max_{i\in[k]} c_{i,j^*,d^*}^2
\sum_{i\in[k]} c_{i,j^*,d^*}^2}\right)
 \leq   2\cdot \tn{exp}\left(-\eps^2/2\nu^2\right) \leq  \eps,
\end{align*}}
for $\nu^2 = \eps^4/4 \leq \eps^2/(2\log(2/\eps))$. Here, step $1$ follows from the Chernoff-Hoeffding inequality (\cref{thm:chernoff}). Since our choice of $\eps <
\xi/(8Td)$, this yields a contradiction
to our choice of $\mc{I}_{-j^*}$,  \eqref{eqn-main-nice-I-1},
and  \eqref{eqn-shared-heavy-def}.

\end{proof}

\subsection{Decoding a Labeling for $\mc{L}$}

\begin{figure}[t]
\begin{mdframed}
\begin{center} \textbf{Randomized Partial Labeling $\sigma$} \end{center}  
\begin{enumerate*}
\item Choose $j^* \in [T]$ and $d^* \in \{0,\dots, d-1\}$
independently and u.a.r.
\item Choose $v_j \in V$ independently and u.a.r. for each $j \in
[T]\setminus\{j^*\}$.
\item Choose the random subset $\mc{I}_{-j^*}$ of $[k]\times ([T]\setminus
\{j^*\})$ by independently adding each element with probability $\eps$.   
\item For each $v \in V$, \\[-1.5em]
\begin{enumerate*}
\item Set $v_{j^*} = v$.
\item Letting $P$ be the restriction of $P_{\tn{global}}$ to 
${\bf Y} = \{Y_{ij}\,\mid\, i\in[k], j \in [T]\}$, define
the set:
\begin{eqnarray}
\Gamma_0(v) & := & \left\{i' \in [k]\,\mid\, c_{i',j^*,d^*}^2 >
\frac{\nu^2}{4}\left(\sum_{i \in [k]} c_{i,j^*,d^*}^2\right)\right\},
\label{eqn-Gamma-0} 
\end{eqnarray} 
where $\nu = \eps^2/4$ (as in Lemma \ref{lem:Azuma-apply}). 
\item If $\Gamma_0(v)$ is non-empty, assign $v$ a label chosen
uniformly at random from $\Gamma_0(v)$.
\end{enumerate*}
\end{enumerate*}
\end{mdframed}
\caption{Randomized Partial Labeling}
\label{fig:labeling}
\end{figure}
In Figure \ref{fig:labeling} we define a randomized (partial) labeling $\sigma$ 
for the vertices $V$ of $\mc{L}$.
To analyze $\sigma$, we first define the
following random subsets of vertices and edges, where the randomness
is over the choices made in the above procedure of labeling.

\medskip
\noindent
{\bf Vertex subset $V_0 \subseteq V$:} Consists of all $v \in V$ such
that:
\begin{itemize}
\item Setting $v_{j^*} = v$, the choice of $\{v_j \,\mid\, j \in
[T]\}$ is good,
\item The choice of $(j^*, d^*)$ satisfied
\eqref{eqn-distinguished-I} and,
\item The choice of $\mc{I}_{-j^*}$ is shared-heavy.
\end{itemize}
Over the randomness of the labeling procedure and a random choice of
$v$, the above happens with probability at least:
\begin{equation}\label{eqn-Delta0}
\Delta_0 := \xi\cdot\frac{1}{Td}\cdot\frac{\xi}{8Td}.
\end{equation}
Thus, $$\E\left[|V_0|\right] \geq \Delta_0|V|.$$ 
Moreover, by the weak
expansion property in Theorem \ref{thm-LC-hardness},
\begin{equation}\label{eqn-E0}
\E\left[|E(V_0)|\right] \geq 
\E\left[\left(\slfrac{|V_0|}{|V|}\right)^2\right]\cdot (|E|/2)
 \geq   \left(\E\left[\slfrac{|V_0|}{|V|}\right]\right)^2\cdot
(|E|/2) 
 \geq  \left(\Delta_0^2/2\right)|E|.
\end{equation}

\medskip
\noindent
{\bf Edge Set $E' \subseteq E(V_0)$:} 
Let us first define for each $v \in V$
\begin{eqnarray}
\Gamma_1(v) & := & \left\{i' \in [k]\,\mid\, c_{i',j^*,d^*}^2 >
\frac{\nu^2}{100\cdot 4^{2R}}\left(\sum_{i \in [k]}
c_{i,j^*,d^*}^2\right)\right\},
\label{eqn-Gamma-1}
\end{eqnarray}
when $v_{j^*}$ is set to $v$ in Step 4a of Figure \ref{fig:labeling}.
Here,  $R$ is the parameter (to be set) from Theorem
\ref{thm-LC-hardness}.
From \eqref{eqn-Gamma-0} and \eqref{eqn-Gamma-1}, we have $\Gamma_0(v) \subseteq \Gamma_1(v)$ along with
\begin{equation}\label{eqn-Gamma-bd}
\left|\Gamma_0(v)\right| \leq 4/\nu^2, \ \ \ \tn{ and } \ \ \ \left|\Gamma_1(v)\right| \leq (100\cdot4^{2R})/\nu^2.
\end{equation}
The set $E'$ is defined as:
\begin{equation}\label{eqn-edgesEprime}
E' := \big\{e = (u,w) \in E(V_0)\,\mid\, \left|\pi_{e,u}(\Gamma_1(u))\right| = \left|\Gamma_1(u)\right|\tn{ and } 
\left|\pi_{e,w}(\Gamma_1(w))\right| = \left|\Gamma_1(w)\right|
\big\}.
\end{equation}
Since the graph $G$ of the instance $\mc{L}$ is regular, using second bound in \eqref{eqn-Gamma-bd} along with 
the smoothness property  of Theorem \ref{thm-LC-hardness}, the fraction of edges $e = (u,w) \in E$ that do not 
satisfy $$\left(\left|\pi_{e,u}(\Gamma_1(u))\right| = \left|\Gamma_1(u)\right|\tn{ and } 
\left|\pi_{e,w}(\Gamma_1(w))\right| = \left|\Gamma_1(w)\right|\right)$$ 
is at most,
$$\Delta_1 := \left(\frac{10^4\cdot 4^{4R}}{\nu^4J}\right).$$
Thus,
\begin{equation}\label{eqn-Eprimebd}
\E\left[\left|E'\right|\right] \geq \left(\Delta_0^2/2 - \Delta_1\right)\left|E\right|.
\end{equation}
The following lemma gives the desired property of edges in $E'$.
\begin{lemma}\label{lem-intersect-proj}
For every edge $e = (u,w) \in E'$, 
\begin{equation}\label{eqn-lem-intesect-proj}
\pi_{e,u}\left(\Gamma_0(u)\right)\cap \pi_{e,w}\left(\Gamma_0(w)\right) \neq \emptyset.
\end{equation}
\end{lemma}
\begin{proof}
Suppose for a contradiction that \eqref{eqn-lem-intesect-proj} does not hold for an edge $e = (u,w) \in E'$, i.e.
\begin{equation}\label{eqn-supposition}
\pi_{e,u}\left(\Gamma_0(u)\right)\cap \pi_{e,w}\left(\Gamma_0(w)\right) = \emptyset.
\end{equation}
Let us now define for $v \in\{u, w\}$, and $i \in [k]$, vector ${\bf C}_{v, i} \in \mathbb{R}^{\mcb{B}_{j^*d^*}}$ where for
any $B \in \mcb{B}_{j^*d^*}$
\begin{equation}
{\bf C}_{v, i}(B) = c_{(i,j^*),B} \ \ \ \tn{ when } v_{j^*}\tn{ is set to }v.
\end{equation}
Without loss of generality, we may assume that 
\begin{equation}\label{eqn-WLOG-umore}
\sum_{i\in [k]} \left\|{\bf C}_{u, i}\right\|_2^2 \ \geq \ \sum_{i\in [k]} \left\|{\bf C}_{w, i}\right\|_2^2.
\end{equation}
Since both $u \in V_0$, \eqref{eqn-main-shared-heavy-I-3} and Lemma \ref{lem:Azuma-apply} imply that there exists $i_u \in [k]$ such that
\begin{equation}\label{eqn-u-distinguished}
\left\|{\bf C}_{u, i_u}\right\|_2 \geq \nu\left(\sum_{i\in [k]} \left\|{\bf C}_{u, i}\right\|_2^2\right)^{\frac{1}{2}} > 0.
\end{equation}
This implies that $i_u \in \Gamma_0(u)$.
Now, let $\ell^* := \pi_{e,u}(i_u)$. Since $P$ is a restriction of $P_{\tn{global}}$ which is a representation of the folded polynomial $\ol{P}_{\tn{global}}$, Lemma \ref{lem-folding-poly}
along with Remark \ref{rem-folding-poly} (applied to elements $B$ of $\mcb{B}_{-j^*d^*}$) implies
\begin{equation}\label{eqn-apply-folding}
\sum_{i \in \pi_{e,u}^{-1}(\ell^*)}{\bf C}_{u, i} = \sum_{i \in \pi_{e,w}^{-1}(\ell^*)}{\bf C}_{w, i}.
\end{equation}
On the other hand, since $e \in E'$, \eqref{eqn-edgesEprime} along with our supposition \eqref{eqn-supposition} and the construction of $\{\Gamma_r(v)\,\mid\, r\in\{0,1\},\ v\in\{u,w\}\}$  
implies that
\begin{itemize}
\item For all $i \in \pi_{e,u}^{-1}(\ell^*)\setminus\{i_u\}$
\begin{equation}\label{eqn-projbd-1}
\left\|{\bf C}_{u, i}\right\|_2 \leq \frac{\nu}{10\cdot 4^R} \left(\sum_{i\in [k]} \left\|{\bf C}_{u, i}\right\|_2^2\right)^{\frac{1}{2}}.
\end{equation}
\item For all $i \in \pi_{e,w}^{-1}(\ell^*)$
\begin{equation}\label{eqn-projbd-2}
\left\|{\bf C}_{w, i}\right\|_2 \leq \frac{\nu}{2} \left(\sum_{i\in [k]} \left\|{\bf C}_{w, i}\right\|_2^2\right)^{\frac{1}{2}}.
\end{equation}
\item There exists at most one $i' \in [k]$ such that,
\begin{equation}\label{eqn-projbd-2}
\left\|{\bf C}_{w, i}\right\|_2 > \frac{\nu}{10\cdot 4^R} \left(\sum_{i\in [k]} \left\|{\bf C}_{w, i}\right\|_2^2\right)^{\frac{1}{2}}.
\end{equation}
\end{itemize}
The above implications along with \eqref{eqn-apply-folding} and \eqref{eqn-WLOG-umore} yields
\begin{eqnarray}\label{eqn-mass-contra}
\left\|{\bf C}_{u, i_u}\right\|_2 & \leq &  \sum_{\substack{i \in \pi_{e,u}^{-1}(\ell^*)\\ i \neq i_u}} \left\|{\bf C}_{u, i}\right\|_2 + \sum_{i \in \pi_{e,w}^{-1}(\ell^*)}\left\|{\bf C}_{w, i}\right\|_2 \nonumber \\
& \leq & \frac{\nu \left|\pi_{e,u}^{-1}(\ell^*)\right|}{10\cdot 4^R} \left(\sum_{i\in [k]} \left\|{\bf C}_{u, i}\right\|_2^2\right)^{\frac{1}{2}} + 
 \left(\frac{\nu}{2} + \frac{\nu \left|\pi_{e,w}^{-1}(\ell^*)\right|}{10\cdot 4^R}\right) \left(\sum_{i\in [k]} \left\|{\bf C}_{w, i}\right\|_2^2\right)^{\frac{1}{2}} \nonumber \\
& \leq & \frac{\nu}{10}\left(\sum_{i\in [k]} \left\|{\bf C}_{u, i}\right\|_2^2\right)^{\frac{1}{2}} + \left(\frac{\nu}{2} + \frac{\nu}{10}\right) \left(\sum_{i\in [k]} \left\|{\bf C}_{w, i}\right\|_2^2\right)^{\frac{1}{2}} \nonumber \\
& \leq & \frac{7\nu}{10}\left(\sum_{i\in [k]} \left\|{\bf C}_{u, i}\right\|_2^2\right)^{\frac{1}{2}},
\end{eqnarray}
where we used the property (from Theorem \ref{thm-LC-hardness}) that $\left|\pi_{e,u}^{-1}(\ell^*)\right|, \left|\pi_{e,w}^{-1}(\ell^*)\right| \leq 4^R$. Clearly, \eqref{eqn-mass-contra}
 is a contradiction to \eqref{eqn-u-distinguished} which completes the proof of the lemma.
\end{proof}
Note that the set $E'$ is determined by Step 3 of the randomized labeling procedure. Lemma \ref{lem-intersect-proj} implies that in the subsequent steps of the procedure, each edge $e = (u,w) \in E'$ is satisfied with probability at least
$$\frac{1}{\left|\Gamma_0(u)\right|\left|\Gamma_0(w)\right|} \geq \frac{\nu^4}{16},$$
using the first bound in \eqref{eqn-Gamma-bd}. The above along with \eqref{eqn-Eprimebd} lower bounds the expected fraction of edges $\sigma$ satisfies by
$$\Delta_2 := \left(\Delta_0^2/2 - \Delta_1\right)\left(\frac{\nu^4}{16}\right).$$
Choosing $R$ to be large enough and $J \gg 4^{4R}$ we can ensure that $\Delta_2 > 2^{-c_0R}$ which yields a contradiction to the soundness of Theorem \ref{thm-LC-hardness}, completing the NO case analysis.

\subsection{Loose Ends}\label{sec:ends}

\paragraph{Discretization of the Basic PCP Test Distribution.}
Let $\mc{H}_N$ be the distribution of $\left(\sum_{i=1}^N B_i\right)/\sqrt{N}$
where each $B_i$ is an independent $\{-1,1\}$-valued balanced
Bernoulli random variable. The following theorem was proved in \cite{DOSW11}.

\begin{theorem} Fix any constant $D \geq 1$, and let $f(x_1,\dots,
x_m)$ be any degree-$D$ polynomial over $\mathbb{R}^m$. Let $(y, z)
\in \mathbb{R}^m\times \mathbb{R}^m$ be generated by sampling each
$(y_i,z_i)$ from $(N(0,1), \mc{H}_N)$ where $N = m^{24D^2}$. Then,
$$\Pr\left[\tn{sign}(f(y)) \neq \tn{sign}(f(z))\right] \leq O(1/m).$$
\end{theorem}

In our Basic PCP Test distribution (for a fixed choice of the vertices of
the \LC instance) we have $m = \Theta(kT)$ Gaussian random
variables. Choosing $D = d$ and $N = m^{24D^2}$, we can completely
discretize the test distribution using $\tn{exp}((kT)^{O(d^2)})$ points. Note
that this also incorporates the possible $2^{O(kT)}$ choices of the
noise set $\mc{I}$. From the above theorem, this discretization
results in an at most $O(1/kT)$ loss in the acceptance probability of
the test. This discretization is done for all
possible choices by the test of the vertices of the instance. 

\paragraph{Ruling out functions of constantly many degree-$d$ PTFs.}
Analogous to the argument in \cite{KS11}, consider any function $\ol{h}$ of
$K$ degree-$d$ PTFs (over $\mc{F}$) that passes the Final PCP test with probability $1/2 +
\xi$. Let $h$ be the function $\ol{h}$ with the PTFs represented over $\mathbb{R}^{\mc{Y}}$. 
By averaging, $h$ flips its sign with respect to flipping $b$ for at
least $\xi$ fraction of the rest of the choices made by the Basic PCP Test.
Again by averaging, there must be a degree-$d$ PTF $\tn{sign}\left(P'_{\tn{global}}\right)$ satisfying the
same for at least $\xi/K$ fraction of the choices. The entire
analysis can then be repeated using $P'_{\tn{global}}$.

\section{Relative bounds for mass in $P$}\label{sec:anticonc}

\newcommand{\Wu}{{\widetilde{\bf U}}}
Let $\mbf{Z}$ denote the set of variables $\{Y_{ij}: j \in [T], k_j < i\leq k\}$. 
As shown in Section \ref{sec-NO-case}, the $\mbf{Z}$ variables are all i.i.d. $N(0,1)$ under the test distribution. 
We begin by expressing $P$ as 
\begin{eqnarray}
P\Big({\bf Z},\{U_i\}_{i \in [T]},\{W_{ij}\}_{i \ne 1}\Big) & =  & P_{\rm omit}\left({\bf Z},\{U_i\}_{i \in [T]},\{W_{ij}\}_{i \in [2,k_j], j\in [T]}\right) + 
Q_0({\bf Z},U_2,\ldots,U_T) \nonumber \\
& & + U_1Q_1({\bf Z},U_1,\ldots,U_T) \label{eqn-split-P-PrelQ0Q1} 
\end{eqnarray}
where $P_{\rm omit}$ consists of all the terms that contain some $\{W_{ij}\,\mid\, i \in [2,k_j], j\in [T]\}$ as a factor, and $Q_0$ is the part in the remaining polynomial independent of $U_1$.
From the nice setting of $\mc{I}$, we have that with probability at least $\xi/2$ over the rest of the choices of the verifier, $P$ flips its sign on flipping $b$. Since $P_{\rm omit}$ evaluates to zero under the test distribution and $Q_0$ is independent of $b\eta$ by construction, we obtain that $Q_1$ is not identically zero. For the time being, our analysis ignores $P_{\rm omit}$.
Extending Definitions \ref{def-hermite-basis-partial} and \ref{def-monomial-Ws-partial}, let $\mcb{H}$ be the Hermite basis over all the ${\bf Z}$ variables, and $\mcb{M}$ be the monomial basis over the variables $\{W_{1j} : j \in [T]\}$. Using these we define two norms to quantify the relevant mass of polynomials.
For convenience, let ${\bf U}$ denote the variables $U_1,\ldots,U_T$, $\Wu$ denote the set ${\bf U} \setminus \{U_1\}$, and ${\bf W}$ denote the set of variables $W_{11},\ldots,W_{1T}$.
\begin{definition}[$\|\cdot\|_2$-norm] 
	Given a polynomial $Q$ over the variables defined in the PCP test, define its $\|\cdot\|_2$-norm as 
	\begin{equation*}
	\|Q\|_2 = \sqrt{\E_{{\bf x} \sim \mc{D}_{\mc{I}}}\Big[|Q({\bf x})|^2\Big]}.
	\end{equation*}
\end{definition}

\begin{definition}[$\|\cdot\|_{\rm mon,1},\|\cdot\|_{\rm mon,2}$-norms]\label{def-monp-norm}
	Given a polynomial $Q({\bf W}) = \sum_{W_S \in \mcb{M}} c_SW_S$ represented in the monomial basis $\mcb{M} = \{W_S\}$, for any $p \ge 1$ define its $\|\cdot\|_{\rm mon,p}$-norm as  
\begin{equation*}
\|Q\|_{\rm mon,p} = \left(\sum_{W_S \in \mcb{M}} |c_S|^p\right)^{1/p}.
\end{equation*}
	In particular, $\|\cdot\|_{\rm mon,1}$ is the absolute sum of the coefficients, and $\|\cdot\|^2_{\rm mon,2}$ is the squared sum of the coefficients in $Q$,
\end{definition}

As pointed out above, $Q_1$ is not identically zero and therefore by definition it satisfies.
\begin{equation}			\label{eq:Q1-nonzero}
	\|Q_1\|_2 > 0
\end{equation}
Our goal in this section is to prove the following lemma lower bounding $\|Q_1\|_2$ relative to $\|Q_0\|_2$.
\begin{lemma}					\label{lem:Q1-lowerBound}
Using the definitions given above,
\begin{equation}
\|Q_0\|_2 \leq \left(\frac{8\eta\sqrt{T}}{(\xi/4d)^{d}\sqrt{\xi}}\right)  \|Q_1\|_2
\end{equation}
\end{lemma}

\begin{proof}
		From Lemma \ref{lem:U-dist}, we know that $U_1 = b\eta\sqrt{T}$ under the distribution $\mc{D}_{\mc{I}}$.  
		Since $Q_1$ is dependent on $U_1$, its distribution can be dependent on $b$. Let $Q^+_1 :=  Q_1|_{b=1}$, and  and $Q^{-}_1 := Q_1|_{b=-1}$. Thus,
		\begin{eqnarray}
		\|Q_1\|^2_2 = \E_{b,Z,{\bf U}}\big[|Q_1|^2\big] &=& \frac{1}{2}\E_{Z,{\bf U}}\Big[|Q_1|^2 \big| b= 1 \Big] + \frac{1}{2}\E_{Z,{\bf U}}\Big[|Q_1|^2 \big| b= -1 \Big] \nonumber \\  &=& \frac{1}{2}\|Q^+_1\|^2_2 + \frac{1}{2}\|Q^-_1\|^2_2. \label{eqn-Q-1-split}
		\end{eqnarray}
		Using the above along with Chebyshev's inequality (see Section \ref{sec-conc-anticonc}) we obtain for any $a > 0$
	\begin{eqnarray}
	\Pr_{\mbf{Z}, \Wu} \left[\left|Q^+_1\right|, \left|Q^-_1\right|  \leq  a\|Q_1\|_2 \right] 
	& \geq & 1 - \Pr\left[\left|Q^+_1\right| \geq  a\|Q_1\|_2 \right] - \Pr\left[\left|Q^-_1\right| \geq  a\|Q_1\|_2 \right] \nonumber \\
	& \geq & 1 - \left(\frac{\|Q^+_1\|^2_2 + \|Q^-_1\|^2_2}{a^2\|Q_1\|_2^2}\right) = 1 - 2/a^2, \label{eqn-both-Q1-small}
	\end{eqnarray}
where the last step follows from   \eqref{eqn-Q-1-split}.
On the other hand note that $Q_0$ is a polynomial over standard Gaussian variables and is independent of $b$. 
Applying the bound of Carbery-Wright (Theorem \ref{thm:carbery-wright-prelim}) 
we obtain the following.
	\begin{equation}
	\Pr\Big[|Q_0| \le (\xi/4d)^d\|Q_0\|_2\Big] \le \frac{\xi}{4} 
	\end{equation}
Setting $a = 4/\sqrt{\xi}$ in   \eqref{eqn-both-Q1-small} and using the above we obtain that with probability at least $1 - \xi/4 - \xi/8 = 1 - 3\xi/8$ over the choice of the variables
$\mbf{Z}$ and $U_2,\dots, U_T$
$$(\eta \sqrt{T})|Q^+_1|, \  (\eta \sqrt{T})|Q^-_1| \ \leq \ (4\eta\sqrt{T/\xi})\|Q_1\|_2, \ \ \ \ \tn{and,} \ \ \ \ |Q_0| \ > \ (\xi/4d)^d\|Q_0\|_2.$$
When $\eta \sqrt{T}(|Q^+_1| + |Q^-_1|) < |Q_0|$ then flipping $b$ does not change the sign of $P$.
Since the sign of $P$ must flip with $b$ with probability at least $\xi/2$ over the choice of $\mbf{Z}$ and $U_2,\dots, U_T$, the above is a contradiction unless,
$$ \|Q_0\|_2 \leq \left(\frac{8\eta\sqrt{T}}{(\xi/4d)^{d}\sqrt{\xi}}\right)  \|Q_1\|_2,$$
which completes the proof of the lemma.
\end{proof}

\section{Proof of Main Structural Lemma \ref{lem-main-nice-I}}\label{sec:findingj}

\newcommand{\Hq}{\widehat{Q}}
\newcommand{\Hp}{\widehat{P}}
\newcommand{\basisA}{{\mcb{B}}}
\newcommand{\basisB}{{\mcb{B}_{-j^*}}}

As in the previous section, we have ${\bf U}$ denote the variables $U_1,\ldots,U_T$, $\Wu$ denote the set ${\bf U} \setminus \{U_1\}$, and ${\bf W}$ denote the set of variables $W_{11},\ldots,W_{1T}$. Similarly, we use $\mbf{Y} = \{Y_{ij} : i \in [k],j \in[T]\}$ to denote the set of all the $Y$ variables. We use $\mbf{Z}$ to denote the set of variables $\{Y_{ij}: j \in [T], k_j < i\leq k\}$. The $\mbf{Z}$ variables are all $N(0,1)$ under the test distribution. For a particular $j^* \in [T]$, let $\mbf{Z}_{j^*} = \mbf{Z} \cap \{Y_{ij^*}: i \in [k]\}$, and let $\mbf{Z}_{-j^*} = \mbf{Z}\setminus \mbf{Z}_{j^*}$. Also, for given $j^*\in [T]$, define $\mbf{Y}_{j^*} = \mbf{Y}\cap \{Y_{ij^*}: i \in [k]\}$ and $\mbf{Y}_{-j^*}= \mbf{Y}\setminus \mbf{Y}_{j^*}$. 
Finally, for given $j^* \in [T]$, we define $\mbf{W}_{j^*}$ and $\mbf{W}_{-j^*}$ similarly.

Recall the definitions of the bases in Definitions \ref{def-hermite-basis-partial}, \ref{def-monomial-Ws-partial} and \ref{def-combined-basis}. Extending these as in the previous section, let $\mcb{H}$ be the Hermite basis for polynomials in the variables $\mbf{Z}$ and $\mcb{M}$ the monomial basis for polynomials in the variables $\mbf{W}$. For any $D \in [d]$, we also define $\mcb{H}_D$ to be the set of all Hermite monomials of degree exactly $D$. 

For convenience of measuring the monomial mass, we use Definition \ref{def-monp-norm} to define two different norms as follows:
 
\begin{definition}[$\|\cdot\|_\basisA$-Norm]				\label{defn:basisA}
For a polynomial $L({\bf Z},{\bf W}) = \sum_{H \in \mcb{H}}H(\mbf{Z}) \cdot L_H({\bf W})$, let
\begin{equation}
\|L({\bf Z},{\bf W})\|^2_\mcb{B} = \sum_{H \in \mcb{H}} \|L_{H}({\bf W})\|^2_{\rm mon,2} 
\end{equation}	
\end{definition}

\begin{definition}[$\|\cdot\|_{\mcb{B}_{-j^*,d^*, J}}$-Norm]\footnote{Note that although we call it so, $\|\cdot\|_{\mcb{B}_{-j^*,d^*,J}}$ is not an actual norm, as it may vanish even for non-zero polynomials.}				\label{defn:norm-BasisB}
Suppose $j^* \in [T], d^* \in [d-1]$ and $J \subseteq [k]$ are given. Then, for any polynomial $M(\mbf{Z}_{-j^*},\mbf{Y}_{j^*}, \mbf{W}_{-j^*})$ of the form
\begin{equation*}
M(\mbf{Z}_{-j^*},\mbf{Y}_{j^*}, \mbf{W}_{-j^*}) = \sum_{H \in \mcb{H}_{-j^*}}\sum_{S \in \mcb{S}_{j^*}} H(\mbf{Z}_{-j^*}) \cdot Y_{S}\cdot M_{H,S}({\bf W}_{-j^*}),
\end{equation*}
we define: 
\begin{equation}
\bigg\|M(\mbf{Z}_{-j^*},\mbf{Y}_{j^*}, \mbf{W}_{-j^*}) \bigg\|^2_{\mcb{B}_{-j^*,d^*, J}} = \sum_{H \in \mcb{H}_{-j^*,d^*}}\sum_{i \in J} \|M_{H,\{(i,j^*)\}}({\bf W}_{-j^*})\|^2_{\rm mon,2} 
\end{equation}
\end{definition}

Finally, for $j^* \in [T]$, we shall find it convenient to define the sets $\mc{A}^{j^*}_1 = \{i : (i,j^*) \in  \mc{I}\}$, and $\mc{A}^{j^*}_0 = [k]\setminus \mc{A}^{j^*}_1$.

\subsection{An intermediate Lemma}
We start by writing the polynomial $P$ in the variables $\mbf{Z}, \{W_{ij}: j \in [T], 1 < i \leq k_j\}, \mbf{U}$:
$$P = P_{\tn{omit}} + P_{\tn{rel}} = P_{\tn{omit}} + \ol{Q}_0(\mbf{Z},\mbf{U}\setminus \{U_1\}) + U_1 \cdot \ol{Q}_1(\mbf{Z}, \mbf{U})$$
where $P_{\tn{omit}}$ contains all monomials depending on variables in $\{W_{ij}: j \in [T], 1 < i \leq k_j\}$. 

Let $Q_0(\mbf{Z}, \mbf{W})$ and $Q_1(\mbf{Z}, \mbf{W})$  be
$\ol{Q}_0$ and $\ol{Q}_1$ respectively after a change of variables from $\mbf{U}$ to $\mbf{W}$. For $a=0,1$, we write $Q_a(\mbf{Z},\mbf{W})$ in the $\mcb{H} \circ \mcb{M}$
basis: $Q_a(\mbf{Z}, \mbf{W}) = \sum_{H \in \mcb{H}} H(\mbf{Z}) \cdot Q_{a,H}(\mbf{W})$. For a fixed $d^* \in \{0\} \cup [d-1]$, we let
$$Q_a^{(d^*)}(\mbf{Z}, \mbf{W}) =  \sum_{H \in \mcb{H}_{d^*}} H(\mbf{Z}) \cdot Q_{a,H}(\mbf{W}).$$

For a fixed $j^* \in [T]$, we define $P_{{\rm omit},j^*}$ as the sub-polynomial of $P$ containing all the monomials containing at least one variable from $\{W_{ij}: j \ne j^*,i \ne 1\}$, and let $P_{{\rm rel},j^*}$ be the rest of the polynomial.

We shall prove Lemma \ref{lem-main-nice-I} using the following intermediate result:
\begin{lemma}					\label{thm:goodj}
There exists choice of $d^* \in \{0,1,\ldots,d-1\}$ and $j^* \in [T]$ such that the following properties hold simultaneously:
	\begin{itemize}
		\item[1.] $\|Q_0\|^2_\basisA \le \rho^{2d}\|Q_1\|^2_\basisA$
		\item[2.] $\|Q_1^{(d^* + 1)}\|^2_\basisA \le \frac14\rho^{d^* + 1}\|Q_1\|^2_\basisA$
		\item[3.] $\Big\|\Wq\Big\|^2_{\mcb{B}_{-j^*,d^*, \mc{A}^{j^*}_0}} \ge \frac{1}{8kT^2}(20dT)^{-4^d}\rho^{d^*}\|Q_1\|^2_\basisA$
	\end{itemize}
	
	\noindent where  $\rho = (20dkT^3/\epsilon^4)^{-6^d}(kT)^{-1}$ and $\Wq({\bf Z}_{-j^*},{\bf Y}_{j^*},{\bf W}_{-j^*})$ is the polynomial obtained by rewriting the ${\bf W}_{j^*}$ variables in $P_{{\rm rel},j^*}$
in terms of the ${\bf Y}_{j^*}$ variables. 
\end{lemma}

Using this, we give a proof of Lemma \ref{lem-main-nice-I}.
\begin{proof}[Proof of Lemma \ref{lem-main-nice-I}]
	Let $d^*$ and $j^*$ be as given in Lemma \ref{thm:goodj}. Let $\Wq({\bf Z}_{-j^*},{\bf Y}_{j^*},{\bf W}_{-j^*},)$ be as in the Lemma \ref{thm:goodj}. We can express $\Wq$ as :	
	\begin{equation}				\label{eqn:Qtilde_defn}
	\Wq({\bf Z}_{-j^*},{\bf W}_{-j^*},{\bf Y}_{j^*}) = \sum_{D = 0}^{d-1}\sum_{H \in \mcb{H}_{-j^*D}}\sum_{S \in \mcb{S}_{j^*}}HY_S\Wq_{H,S}({\bf W}_{-j^*})
	\end{equation}

	\noindent where $\mcb{H}_{-j^*D}$ is the set of Hermite monomials which are of degree $D$ and do not contain ${\bf Z}_{j^*}$ variables. By construction we have
	\begin{equation}					\label{eq:c_sum1}
	\sum_{(i,j^*) \in \mc{I}} c_{i,j^*,d^*}^2  = \| \Wq\|^2_{\mcb{B}_{-j^*,d^*,\mc{A}^{j^*}_1}}
	\end{equation}

	Consider a term that contributes to the RHS of \eqref{eq:c_sum1} (as defined in \ref{defn:norm-BasisB}). Since the additional $Y_{ij^*}$ (for $(i,j^*) \in \mc{I}$) variable adds to the degree of $H$, the corresponding term appears in the $\basisA$-representation of $P_{\rm rel}$ as $HM$ where the degree of $H$ is of degree $d^*+1$. Therefore it must be a part of $Q^{(d^*+1)}_0$ or $Q^{(d^* + 1)}_1$. Hence,
	 \begin{equation}
	 	\| \Wq\|^2_{\mcb{B}_{-j^*,d^*,\mc{A}^{j^*}_1}}   
	 	\le \|U_1{Q}^{(d^* + 1)}_1\|^2_\basisA + \|{Q}_0\|^2_\basisA  
	 	\overset{1}{\le}     T\|{Q}^{(d^* + 1)}_1\|^2_\basisA + \rho^{2d}\|Q_1\|^2_\basisA 
	 	\le 2T\rho^{d^* + 1}\|Q_1\|^2_\basisA			\label{eq:c_sum3} 		
	 \end{equation}
    \noindent where the upper bound on the first term in step $1$ follows from  
	\begin{eqnarray*}
	\|U_1({\bf W}){Q}^{(d^* + 1)}_1({\bf W})\|^2_\basisA 
	&=&\sum_{H \in \mcb{H}_{d^* + 1}} \|U_1({\bf W}) Q_H({\bf W})\|^2_{\rm mon,2} \\
	&\le& \sum_{H \in \mcb{H}_{d^* + 1}} \|U_1({\bf W})\|^2_{\rm mon,1} \|Q_H({\bf W})\|^2_{\rm mon,2}  \qquad\qquad \Big(\mbox{Claim } \ref{cl:compare_masses2}\Big)\\
	&=&T\|Q^{(d^* + 1)}_1\|^2_\basisA
	\end{eqnarray*}
	\noindent and the upper bound on the second term in step $1$ follows from Lemma \ref{thm:goodj} (part 1). The last inequality uses Part 2. of Lemma \ref{thm:goodj}. On the other hand we have,
	\begin{equation}				\label{eq:csum_2}					
		\sum_{(i,j^*) \in ([k]\times\{j^*\})\setminus\mc{I}} c_{i,j^*,d^*}^2 = \|\Wq\|^2_{\mcb{B}_{-j^*,d^*,\mc{A}^{j^*}_0}} 
	\end{equation}
	From Lemma \ref{thm:goodj} (part 3) and the choice of $\rho$ in Lemma \ref{thm:goodj} we have
	\begin{equation}			\label{eq:c_sum4}
	\|\Wq\|^2_{\mcb{B}_{-j^*,d^*,\mc{A}^{j^*}_0}}  \ge \frac{1}{8kT^2}(20dT)^{-4^d}\rho^{d^*}\|Q_1\|^2_\basisA \ge \frac{16T}{\epsilon^4}\rho^{d^* + 1}\|Q_1\|^2_\basisA	
	\end{equation}
	Combining   \eqref{eq:c_sum3},\eqref{eq:csum_2} and \eqref{eq:c_sum4}, we get an upper bound on LHS of \eqref{eq:c_sum1} which gives us
	\begin{equation}
		\sum_{(i,j^*) \in \mc{I}} c_{i,j^*,d^*}^2  \le \frac{\epsilon^4}{8} \bigg(\sum_{(i,j^*) \in ([k]\times\{j^*\})\setminus\mc{I}} c_{i,j^*,d^*}^2\bigg)
	\end{equation}
	thus implying inequality \eqref{eqn-main-nice-I-1}. Furthermore, from \eqref{eq:Q1-nonzero}, we know that $\|Q_1\|^2_2 > 0$, which along with Lemma \ref{lem:coeff_bounds}(part 1) implies that $\|Q_1\|^2_\basisA > 0$. Therefore, combining \eqref{eq:c_sum4} and \eqref{eq:csum_2}, we get that the LHS of \eqref{eq:csum_2} is strictly positive, thus implying \eqref{eqn-main-nice-I-3}. Hence, the choice of $(d^*,j^*)$ satisfy \eqref{eqn-main-nice-I-1} and \eqref{eqn-main-nice-I-3}.
\end{proof}

\subsection{Proof of Lemma \ref{thm:goodj}}

\subsubsection{Upper bounding $\|Q_0\|_\basisA$ in terms $\|Q_1\|_\basisA$}		\label{sec:Q0Q1-comparison}
  In this section, we show that $\|Q_0\|_\basisA$ is small compared  terms $\|Q_1\|_\basisA$ due to our choice of $\eta$. 
	\begin{lemma}				\label{lem:Q0-Q1-bound}
		Let $\rho$ be chosen as in Lemma \ref{thm:goodj}. Then $\|Q_0\|^2_\basisA \le \rho^{2d}\|Q_1\|^2_\basisA$
	\end{lemma}
  \begin{proof}
	We express $Q_0$ as 
	\begin{equation*}
	Q_0({\bf Z},{\bf W}) = \sum_{H \in \mcb{H}} H Q_{0,H}({\bf W})
	\end{equation*}
	where $H \in \mcb{H}$ are the Hermite monomials. Then by definition of $\|\cdot\|^2_\basisA$ we have,
	\begin{eqnarray*}
		\|Q_0({\bf Z},{\bf W})\|^2_\basisA &=& \sum_{H \in \mcb{H}} \|Q_{0,H}({\bf W})\|^2_{\rm mon,2} \\
		&\overset{1}{\le} & (10dT)^{14d}\sum_{H \in \mcb{H}} \|Q_{0,H}(\Wu)\|^2_{2} \\
		&{=} & (10dT)^{14d}\|Q_0(\Wu)\|^2_{2} \\
		&\overset{2}{\le} & \frac{\rho^{4d}}{4}\|Q_1\|^2_{2} 
	\end{eqnarray*}
	\noindent where step $1$ follows from Lemma \ref{lem:coeff_bounds} (part $2$), and step $2$ follows from Claim \ref{lem:Q1-lowerBound} and our choice of $\eta$ in Section \ref{sec:reduction}. Furthermore, we can relate the $\|Q_1\|^2_2$ to $\|Q_1\|^2_\basisA$ as follows
	\begin{equation*}				
	\|Q_1\|^2_\basisA = \sum_{H \in \mcb{H}}\|Q_{1,H}({\bf W})\|^2_{\rm mon,2} \overset{1}{\ge} (20dT)^{-10d}\sum_H\|Q_{1,H}({\bf U})\|^2_2
	=  (20dT)^{-10d}\|Q_1\|^2_2 
	\end{equation*}
	\noindent where step $1$ follows from Lemma \ref{lem:coeff_bounds} (part $1$). Combining the bounds, we get $\|Q_0\|^2_\basisA \le \rho^{2d}\|Q_1\|^2_\basisA$.
	\end{proof}
\subsubsection{Finding a heavy $d^* \in \{0,1,\ldots,d-1\}$}

We begin by finding a $d^* \in \{0\} \cup [d-1]$ such that $Q_1$ restricted to Hermite monomials in $\mcb{H}_{d^*}$ has large mass compared to those from $\mcb{H}_{d^* + 1}$.

\begin{lemma}				\label{lem:heavy-d*}
	There exists $d^* \in \{0\} \cup [d-1]$ such that 
	\begin{itemize}
		\item[1.]  $\|Q^{(d^*+1)}_1\|^2_\basisA \le \frac14\rho^{d^*+1}\|Q_1\|^2_\basisA$
		\item[2.] $\|Q^{(d^*)}_1\|^2_\basisA \ge \frac14\rho^{d^*}\|Q_1\|^2_\basisA$
	\end{itemize}
\end{lemma}
\begin{proof}
    We claim that there exists $D \in \{0\} \cup[d-1]$ such that $ \|Q^{(D)}_1\|^2_\basisA \ge \frac14\rho^{D}\|Q_1\|^2_\basisA$. If not, then for all $D \in \{0\} \cup [d-1]$ we have $ \|Q^{(D)}_1\|^2_\basisA <  \frac14\rho^{D}\|{Q}_1\|^2_\basisA$.  Then,
    \begin{equation*}
    \|Q_1\|^2_\basisA =  \sum_{D= 0}^{d-1} \|Q^{(D)}\|^2_\basisA \le \sum_{D= 0}^{d-1} \frac{\rho^{D}}{4}\|Q_1\|^2_\basisA < \frac{1}{2}\|Q_1\|^2_\basisA
    \end{equation*}
    \noindent which is a contradiction. 	

	Now we set $d^*$ to be the largest such $D \in \{0\} \cup [d-1]$ such that $\|Q^{(D)}_1\|^2_\basisA \ge \frac14 \rho^D\|Q_1\|^2_\basisA$. If $d^* < d-1$, then by construction we know that $\|Q^{(d^* + 1)}_1\|^2_\basisA < \frac14\rho^{d^* + 1}\|Q_1\|^2_\basisA$. On the other hand if $d^* = d - 1$, then by construction $Q^{(d^* + 1)}_1$ is identically $0$ (since $Q_1$ is of degree at most $d-1$) and hence the claim is vacuously true. 
\end{proof}

\subsubsection{Locating a good $j^* \in [T]$}

Let  $d^* \in \{0\} \cup [d-1]$ be as in Lemma \ref{lem:heavy-d*}. Now, we shall find a good $j^* \in [T]$ in the sub-polynomial $U_1Q_1^{(d^*)}$ which contains a sub-polynomial linear in $W_{1j^*}$ with significant $\|\cdot\|_\basisA$-mass.

\begin{lemma}					\label{lem:goodj}
Let the polynomial $U_1Q^{(d^*)}({\bf Z},{\bf W})$ be expressed in the basis $\mcb{B}$ as 
\begin{equation*}
U_1Q^{(d^*)}({\bf Z},{\bf W}) = \sum_{H \in \mcb{H}_{d^*}}\sum_{M \in \mcb{M}} c_{H,M}HM 
\end{equation*}
Then there exists $j^* \in [T]$ such that 
\begin{equation}							\label{eq:goodj-main}
\sum_{H \in \mcb{H}_{-j^*d^*}}\sum_{M \in \mcb{M}_{-j^*}} c^2_{H,MW_{1j^*}} \ge \frac{1}{T^2}(20dT)^{-4^d}\bigg(\sum_{H \in \mcb{H}_{d^*}}\sum_{M \in \mcb{M}} c^2_{H,M}\bigg)
\end{equation}
\end{lemma}
\begin{proof}
Consider the following representation of $U_1Q^{(d^*)}_1$:
\begin{equation}				\label{eq:U1Q1-decomp}
U_1Q^{(d^*)}_1({\bf Z},{\bf W}) = \sum_{H \in \mcb{H}_{d^*}}HU_1Q_{1,H}({\bf W})
\end{equation}
Using the fact that $U_1 = (1/\sqrt{T})\sum_{j=1}^TW_{1j}$ and $T = 10d$, the following lemma is directly implied by Lemma \ref{lem:robust_polynomial}.
\begin{lemma}				\label{lem:robust_poly_restated}
	Fix $H \in \mcb{H}_{d^*}$. Let $U_1Q_{1,H}({\bf W})$ (as defined in   \eqref{eq:U1Q1-decomp}) be expressed in the basis $\mcb{B}$ as 
	\begin{equation*}
	U_1Q_{1,H}({\bf W}) = \sum_{M \in \mcb{M}}c_{H,M}M
	\end{equation*}
	Then there exists at least $T/2$ choices of $j^* \in [T]$ such that  
	\begin{equation}					\label{eq:j*-mass}
	\sum_{M \in \mcb{M}_{-j^*}}c^2_{H,MW_{1j^*}} \ge  \frac{1}{T}(20dT)^{-4^d}\sum_{M \in \mcb{M}}c^2_{H,M}
	\end{equation}
\end{lemma} 
For a fixed Hermite monomial $H \in \mcb{H}_{d^*}$, we call a $j^* \in [T]$ to be \emph{good} for $H$ if the following conditions hold:
\begin{itemize}
\item[1.] The Hermite monomial $H$ does not contain ${\bf Z}_{j^*}$-variables.
\item[2.] The index $j^*$ satisfies   \eqref{eq:j*-mass} with respect to $H$
\end{itemize}
Now for a fixed Hermite monomial $H \in \mcb{H}_{d^*}$, out of $T$ values of $j$, at most $d-1$ can appear in $H$. Furthermore, Lemma \ref{lem:robust_poly_restated} guarantees that for at least $T/2$-values of $j \in [T]$,   \eqref{eq:j*-mass} is satisfied. Since $T = 10d$, for each Hermite monomial $H$ there exists at least some $ j^*(H)$ which is good for $H$. Therefore by averaging over all $H \in \mcb{H}_{d^*}$, there exists $j^* \in [T]$ such that 
\begin{equation*}		
\sum_{H \in \mcb{H}_{-j^*d^*}}\sum_{M \in \mcb{M}_{-j^*}} c^2_{H,MW_{1j^*}} \ge \frac{1}{T}\sum_{H \in \mcb{H}_{-d^*}}\sum_{M \in \mcb{M}_{-j^*(H)}}c^2_{H,MW_{1j^*(H)}} \ge \frac{1}{T^2}(20dT)^{-4^d}\bigg(\sum_{H \in \mcb{H}_{d^*}}\sum_{M \in \mcb{M}} c^2_{H,M}\bigg)
\end{equation*}

\end{proof}

\subsubsection{Substituting ${\bf W}_{j^*}$ with ${Y}_{j^*}$-variables} 

For the $j^* \in [T]$ chosen in the previous section, $P_{{\rm rel},j^*}$ can be rewritten by expanding ${\bf W}_{j^*}$ in the ${\bf Y}_{j^*}$-variables as $\Wq({\bf Z}_{-j^*},{\bf Y}_{j^*},{\bf W}_{-j^*})$ which can be expressed in the basis $\mcb{B}_{-j^*}$ as follows:
\begin{equation}				\label{eq:expanded_Qtilde}				
\Wq({\bf Z}_{-j^*},{{\bf W}_{- j^*}},{\bf Y}_{j^*}) = \sum_{D = 0}^{d-1}\sum_{H \in \mcb{H}_{-j^*D}}\sum_{M \in \mcb{M}_{-j^*}}\sum_{S \in \mcb{S}_{j^*}}\tilde{c}_{H,M,S}HMY_S 
\end{equation}
	where $\mcb{H}_{-j^*D}$,$\mcb{M}_{-j^*}$ and $\mcb{S}_{j^*}$ are as defined in Section \ref{sec:hybrid}. Now we show that the squared sum of coefficients in the above expression, restricted to factors to terms of the form $HMY_{ij^*}$ capture a significant fraction of mass. 
\begin{cl}					\label{cl:W-to-Y}
	Let $\Wq({\bf Z}_{-j^*},{{\bf W}_{- j^*}},{\bf Y}_{j^*})$ be as in   \eqref{eq:expanded_Qtilde}. Then,
	\begin{equation}					
	 \sum_{H \in \mcb{H}_{-j^*d^*} }\sum_{M \in \mcb{M}_{-j^*}}\sum_{i \in [k_{j^*}]}\tilde{c}^2_{H,M,({ij^*})} \ge \frac{1}{2k_{j^*}}\bigg(\sum_{H \in \mcb{H}_{-j^*d^*}}\sum_{M \in \mcb{M}_{-j^*}} c^2_{H,MW_{1j^*}}\bigg)				\label{eq:Qtilde-lower-bound}
	\end{equation}
\end{cl}
\begin{proof}
	Consider the polynomial $P_{\rm lin}$ defined as follows:
	\begin{equation}				\label{eq:expand-P-lin}
	P_{\rm lin}({\bf Z},{\bf W}) = \sum_{H \in \mcb{H}_{-j^*d^*}}\sum_{M \in \mcb{M}_{-j^*}}\sum_{i \in [k_{j^*}]}\alpha_{H,M,i}HMW_{ij^*}
	\end{equation}
	\noindent which is the sub-polynomial in $P$ consisting of monomials containing exactly one ${\bf W}_{j^*}$-variable. Note that terms on the RHS of \eqref{eq:expand-P-lin} for $i > 1$ 
are contained in $P_{\tn{omit}}$.    

	Fix a $HM \in \mcb{H}_{-j^*d^*} \circ \mcb{M}_{-j^*}$ and $i \in [k_{j^*}]$. Under the linear transformation ${\bf W}_{j^*} \mapsto {\bf Y}_{j^*}$ we have
	\begin{equation}
	\tilde{c}_{H,M,({ij^*})} = \sum_{ l \in [k_{j^*}] }\alpha_{H,M,l}c_{l,i}
	\end{equation}
	where the $c_{1,l},\ldots,{c}_{T,l}$ are the $l^{th}$ coordinates of vectors ${\bf c}_1,\ldots,{\bf c}_T$ (as in Section \ref{sec-NO-case}). Recall that $\langle {\bf c}_i,{\bf c}_{i^\prime} \rangle = 0$ for all $i \ne i^\prime$. Therefore
	\begin{eqnarray}
	\sum_{i \in [k_{j^*}]} \tilde{c}^2_{H,M,Y_{ij^*}} &=& \Big\|\sum_{ l \in [k_{j^*}] }\alpha_{H,M,l}{\bf c}_{l}\Big\|^2   \\
											   &=& \sum_{ l \in [k_{j^*}] }\Big\|\alpha_{H,M,l}{\bf c}_{l}\Big\|^2   \\
											   &\ge& \alpha^2_{H,M,1}\|{\bf c}_{1}\|^2 = \frac{\alpha^2_{H,M,1}}{k_{j^*}} 
	\end{eqnarray}
	
	To finish the proof, we note that for $i = 1$ the RHS of \eqref{eq:expand-P-lin} has contribution either from terms in $U_1Q_1^{(d^*)}$ or $Q_0$. 
	 Summing over all pairs $HM \in \mcb{B}_{-j^*}$ and using the triangle inequality we obtain 
	
	\begin{eqnarray}
	\sqrt{\sum_{H \in \mcb{H}_{-j^*d^*}}\sum_{M \in \mcb{M}_{-j^*} }\alpha^2_{H,M,1}} &\ge& \sqrt{\sum_{H \in \mcb{H}_{-j^*d^*}}\sum_{M \in \mcb{M}_{-j^*}} c^2_{H,M,W_{1j^*}}} - \|Q_0\|_{\basisA} \\
	&\ge& \frac{1}{\sqrt{2}}\sqrt{\sum_{H \in \mcb{H}_{-j^*d^*}}\sum_{M \in \mcb{M}_{-j^*}} c^2_{H,M,W_{1j^*}}}  
	\end{eqnarray}
	\noindent where we upper bound $\|Q_0\|_{\basisA}$ as follows:
	\begin{eqnarray}
	\|Q_0\|^2_{\basisA} \overset{1}{\le} \rho^{2d}\|Q_1\|^2_{\basisA} \overset{2}{\le} \rho^{d}\|Q^{(d^*)}_1\|^2_{\basisA} 
	&=& \rho^d\sum_{H \in \mcb{H}_{d^*}}\sum_{M \in \mcb{M}} c^2_{H,M} \\
	&\overset{3}{\le}& \frac{1}{16}   \sum_{H \in \mcb{H}_{-j^*d^*}}\sum_{M \in \mcb{M}_{-j^*}} c^2_{H,M,W_{1j^*}}  
	\end{eqnarray}
	\noindent where inequality $1$ follows from Lemma \ref{lem:Q0-Q1-bound}, inequality $2$ follows from Lemma \ref{lem:heavy-d*} and the last inequality follows from Lemma \ref{lem:goodj} and our choice of $\rho$.
\end{proof}

\subsubsection{Completing the proof of Lemma \ref{thm:goodj}}

	Part $1$ follows from Lemma \ref{lem:Q0-Q1-bound} and Part $2$ follows directly from Lemma \ref{lem:heavy-d*}. For Part $3$, observe that the LHS of Part $3$ (in Lemma \ref{thm:goodj}) is equal to the LHS of   \eqref{eq:Qtilde-lower-bound}, which can be lower bounded using Claim \ref{cl:W-to-Y}, Lemma \ref{lem:goodj} and Lemma \ref{lem:heavy-d*} as follows
	\begin{eqnarray}
	\frac{1}{2k_{j^*}}\bigg(\sum_{H \in \mcb{H}_{-j^*d^*}}\sum_{M \in \mcb{M}_{-j^*}} c^2_{H,MW_{1j^*}}\bigg) &\ge& \frac{1}{2k_{j^*}T^2}(20dT)^{-4^d}\bigg(\sum_{H \in \mcb{H}_{d^*}}\sum_{M \in \mcb{M}} c^2_{H,M}\bigg) \\
	 &=& \frac{1}{2T^2k_{j^*}}(20dT)^{-4^d}\|Q^{(d^*)}_1\|^2_\basisA \\
	 &\ge& \frac{1}{8T^2k_{j^*}}(20dT)^{-4^d}\rho^{d^*}\|Q_1\|^2_\basisA 
	\end{eqnarray}
	
	which completes the proof.

\section{A Linear Mass Bound for Low Degree Polynomials}\label{sec:struct}

In this section we study the structure of polynomials over the variable set $\{W_1,\dots, W_T\}$. For a polynomial $P(W_1,\dots,W_T)$, dropping the subscript we use $\|P\|$ to denote the $\ell_2$-norm of the coefficients of $P$ in the monomial basis.  
Let $U := \sum_{j=1}^T W_j.$ Define $Q(W_1, \dots, W_T) = U \cdot S(W_1, \dots, W_T)$, a polynomial of degree $d+1$. For any $j \in [T]$, write:
		\begin{align}
			S(W_1,\ldots,W_T) &= \sum_{\ell = 0}^{d}W^\ell_j \cdot S_{j,\ell}({\bf W}_{\ne j}) \label{eq:expand_S}\\
			Q(W_1,\ldots,W_T) &= \sum_{\ell = 1}^{d+1}W^\ell_j \cdot Q_{j,\ell}({\bf W}_{\ne j}) 		\label{eq:expand_Q}
		\end{align}
		where $\mb{W}_{\neq \sigma}= \{W_i\}_{i \notin \sigma}$ for any list $\sigma$ of indices.
The main result of this section is the following lemma showing that for many $j \in [T]$, the $W_j$-linear sub-polynomial $Q_{j,1}$ has significant mass:
\begin{lemma}					\label{lem:robust_polynomial}
	For polynomials $S$ and $Q$ as above, if $T > 2d$, there are at least $T/2$ choices of $j \in [T]$ such that $\|Q_{j,1}\|\geq (20dT)^{-3^d}\|S\|$. 
\end{lemma}
The rest of this section is devoted to proving Lemma \ref{lem:robust_polynomial}.

\subsection{The Variable Removal Lemma}

The key ingredient that is needed to prove this is the following lemma that will be iteratively applied while reducing the number of variables and the degree at each iteration:
\begin{lemma}[Variable Removal]					\label{lem:varred}
Let $d\geq 1$. For variables $X, Y, Z$, suppose there are polynomials $S_1, S_2$ of degree $d-1$, polynomials $R_1, R_2$ of degree $d-2$, and error polynomials $\Delta^X,\Delta^Y$ of degree $d$ satisfying:
    \begin{eqnarray}	
	& & (aX- Y- Z)S_1(Y,Z) + \Delta^X(X,Y,Z) + X^2R_1(X,Y,Z) \nonumber \\ 
	& = & (aY- X- Z)S_2(X,Z) + \Delta^Y(X,Y,Z) + Y^2R_2(X,Y,Z). \label{eqn:rp}
	\end{eqnarray}
Then,
$$S_1(Y,Z) = \Big((a+1)Y - Z \Big)C(Z) +Y^2A_{1}(Y,Z)+\Delta(Y,Z)$$ 
where $\Delta$ is such that $\|\Delta\| \le 20a\max(\|\Delta^X\|,\|\Delta^Y\|)$. Furthermore, we have $\deg(C(Z)) \le d-2$, $\deg(A_1(Y,Z)) \leq d-3$, and $\deg(\Delta(Y,Z))\leq d-1$. 
\end{lemma}

\begin{proof}
We write the polynomials $S_1$ and $S_2$ in the following way\footnote{If $d\leq 2$, then some of the polynomials below are automatically $0$.}:
	\begin{eqnarray*}
	S_1(Y,Z) &=& Y^2 \cdot A_{1}(Y,Z) + Y \cdot B_1(Z) + Z \cdot C_1(Z) + D_1 \\
	S_2(X,Z) &=& X^2\cdot A_{2}(X,Z) + X\cdot B_2(Z) + Z\cdot C_2(Z) + D_2
	\end{eqnarray*}
Note that $C_1(Z)$ and $A_1(Y,Z)$ can be of degree at most $d-2$ and $d-3$ respectively. Additionally, we write the error polynomials as:	
	\begin{eqnarray*}
	\Delta^X &=& X\cdot \Delta^X_X + Z\cdot \Delta^X_Z + Z^2\cdot \Delta^X_{Z^2}(Z) + YZ\cdot \Delta^X_{YZ}(Z) + \tilde{\Delta}^X(X,Y,Z) \\
	\Delta^Y &=& X\cdot \Delta^Y_X + Z\cdot \Delta^Y_Z + Z^2\cdot \Delta^Y_{Z^2}(Z) + YZ\cdot \Delta^Y_{YZ}(Z) + \tilde{\Delta}^Y(X,Y,Z)
	\end{eqnarray*}
To be clear, the functions without any arguments, such as $\Delta^X_X$ or $\Delta^Y_Z$, are constants. The above decomposition is unique.
 Now we match coefficients in \eqref{eqn:rp}. 
	
	\begin{itemize}
		\item[1.] Matching terms of the form $X^0Y^0Z^{\ge 2}$, we get  $-C_1(Z) + \Delta^X_{Z^2} = -C_2(Z) + \Delta^Y_{Z^2} \Rightarrow C_2(Z) = C_1(Z) + \Delta^Y_{Z^2} - \Delta^X_{Z^2}$
		\item[2.] Matching terms of the form $X^1Y^0Z^0$, we get $aD_1 + \Delta^X_{X} = -D_2 + \Delta^Y_X \Rightarrow D_2 = - aD_1 + \Delta^Y_X - \Delta^X_X$
		\item[3.] Matching terms of the form $X^0Y^0Z^1$, we get $-D_1 + \Delta^{X}_Z= -D_2 + \Delta^{Y}_Z$. Substituting $D_2$ from above:
		\begin{align*}
		-D_1 &= - D_2  + \Delta^Y_Z - \Delta^X_Z
		     = aD_1 - (\Delta^Y_X - \Delta^X_X)  + (\Delta^Y_Z - \Delta^X_Z) 
		\end{align*}
		which on rearranging gives us
		$D_1 = -\frac{1}{a+1}\Big[ \Delta^Y_Z - \Delta^X_Z - \Delta^Y_X + \Delta^X_X\Big]$
\item[4.] Matching $X^0Y^1Z^{\ge 1}$ we get $-B_1(Z) - C_1(Z) + \Delta^X_{YZ} = aC_2(Z) + \Delta^Y_{YZ}$. Substituting $C_2(Z)$ from above,
		\begin{eqnarray*}
		-B_1(Z) &=& aC_2(Z) + C_1(Z) + \Delta^Y_{YZ} - \Delta^X_{YZ}  \\
				&=& a\Big(C_1(Z) + \Delta^Y_{Z^2} - \Delta^X_{Z^2}\Big) + C_1(Z) + \Delta^Y_{YZ} - \Delta^X_{YZ}  \\	
				&=& (a + 1)C_1(Z) + a\big(\Delta^Y_{Z^2} - \Delta^X_{Z^2}\big) + \Delta^Y_{YZ} - \Delta^X_{YZ} 	
		\end{eqnarray*}
	\end{itemize}
	Finally by substituting $B_1(Z)$ and $D_1$ in the expression for $S_1(Y,Z)$ and collecting the error terms, we get
	\begin{align*}
	S_1(Y,Z) &= Y^2A_1(Y,Z) - Y\Big[(a + 1)C_1(Z) + a\big(\Delta^Y_{Z^2} - \Delta^X_{Z^2}\big) + \Delta^Y_{YZ} - \Delta^X_{YZ}\Big] + ZC_1(Z) + D_1 \\
	&=   Y^2A_1(Y,Z) -C_1(Z)\Big[(a+1)Y - Z\Big] - Y\Big[a\big(\Delta^Y_{Z^2} - \Delta^X_{Z^2}\big) + \Delta^Y_{YZ} - \Delta^X_{YZ}\Big]\\ 
	&- \frac{1}{a+1}\Big[\Delta^Y_Z - \Delta^X_Z - \big(\Delta^Y_X - \Delta^X_X \big)\Big]\\
	&=   Y^2A_1(Y,Z) -C_1(Z)\Big[(a+1)Y - Z\Big] + \Delta(Y,Z)
	\end{align*}
We obtain the lemma setting $C(Z) = -C_1(Z)$ and $\Delta(Y,Z) = - Y\Big[a\big(\Delta^Y_{Z^2} - \Delta^X_{Z^2}\big) + \Delta^Y_{YZ} - \Delta^X_{YZ}\Big]$.
The upper bound on $\|\Delta\|$ follows by triangle inequality.		
	\end{proof}
	
		\subsection{Proof of Lemma \ref{lem:robust_polynomial}}

 Fix $j \in [T]$. Comparing the coefficients of the sub-polynomial that are degree $1$ in $W_j$ in the expansion of $Q$ (see  \eqref{eq:expand_Q}) and $US$ (see \eqref{eq:expand_S}), we get
		\begin{equation}							\label{eq:Delta}
			 Q_{j,1}({\bf W}_{\ne j}) = S_{j,0} + S_{j,1}({\bf W}_{\ne j})\sum {\bf W}_{\ne j}
		\end{equation}
where $\sum \mb{W}_{\neq\sigma}$ is the sum of all variables in $\mb{W}_{\neq\sigma}$ for any list $\sigma$ of indices.		
Denote $Q_{j,1}({\bf W}_{\ne j})$ as $\Delta^{(1)}_j$.

		The proof is by contradiction i.e., we assume that more than $T/2$ of the $\Delta^{(1)}_j$ polynomials have small mass. We show first that there exist many $j$'s such that the sub-polynomial of $S$ not divisible by $W_j^2$ retains significant mass. This is achieved using Lemma \ref{lem:lower_bound}. Next, we apply Lemmas \ref{lem:varred} and \ref{lem:lower_bound} as well as the degree bound on $S$ to obtain a contradiction.
	
	\subsubsection{Finding a non-quadratic sub-polynomial with significant mass}
	\begin{lemma}				\label{lem:lower_bound}
		Given a polynomial $P$ on variables $W_1,\ldots,W_T$ of degree $d$ such that $\|P\| = 1$,  let $P = W_j^2 P_j(W_1,\ldots,W_T) + R_j(W_1,\ldots,W_T)$ for every $j \in [T]$ where $R_j(\cdot)$ is the sub-polynomial which does not contain a $W_j^2$ factor. Then, if $T > d$, there exists $j \in [T]$ such that $\|R_j\| > 4^{-2^d}\|P\|$.
	\end{lemma}
	\begin{proof}
Without loss of generality, assume $\|P\| = 1$ by rescaling.
Suppose that for all $j \in [d]$, $\|R_j\| \leq \eta \coloneqq 4^{-2^d}$. We show that this violates the degree bound on $P$ using the following claim.
\begin{claim}\label{clm:deg}
For every $j \in [d]$, if polynomials $H_j$ and $L_j$ are defined such that $P = W_1^2 \cdots W_j^2 \cdot H_j + L_j$ and $L_j$ is not divisible by $W_1^2 \cdots W_j^2$, then $\|L_j\|\leq 4 \cdot \eta^{1/2^{j-1}}$.
\end{claim}
This claim proves the lemma because it shows $\|L_d\| \leq 4 \cdot \eta^{1/2^{d-1}} < 1/2$, so $\|H_d\|>0$ (since they contribute disjoint monomials to $P$), and therefore $P$ contains a monomial of degree $2d$, a contradiction.
	\end{proof}
	
\begin{proof}[Proof of Claim \ref{clm:deg}]
The proof is by induction on $j$. The base case $j=1$ is clear, since $L_1 = R_1$.

For the inductive step, suppose the claim is true for $j-1$. Then, we have that $W_j^2 P_j + R_j = P = W_1^2 \cdots W_{j-1}^2 H_{j-1} + L_{j-1}$ with $\|L_{j-1}\| \leq 4\eta^{1/2^{j-2}}$. Write $H_{j-1} = W_j^2 H_j' + L_j'$ where $L_j'$ is not divisible by $W_j^2$. Now, $P = W_1^2 \cdots W_j^2 H_j' + W_1^2 \cdots W_{j-1}^2 L_j' + L_{j-1}$.

By looking at the terms divisible by $W_j^2$, we have that $\|W_j^2 P_j\| = \|P_j\| \leq \|H_j'\| + \|L_{j-1}\|$. Since $\|P_j\| \geq 1-\eta$ and $\|L_{j-1}\| \leq 4 \eta^{1/2^{j-2}}$, we get that $\|H_j'\| \geq 1-8\eta^{1/2^{j-2}}$.  

 Let $H_j = H_j'$ and $L_j = W_1^2 \cdots W_{j-1}^2 L_j' + L_{j-1}$. Then, 
\begin{align*}
\|L_j\|^2 = 1-\|H_j\|^2 &= 1-\|H'_j\|^2 \leq 1-(1-8\eta^{1/2^{j-2}})^2 \leq 16 \eta^{1/2^{j-2}}
\end{align*}
\end{proof}
	\subsubsection{Iterative expansion of $S$}
We are now ready to prove Lemma \ref{lem:robust_polynomial}. For contradiction, suppose that $\max_{j \in [T/2]}\|\Delta^{(1)}_j\| \le C_{\max} \coloneqq (20dT)^{-3^d}$. By rescaling, we can assume $\|S\| = 1$. We expand the polynomial $S$ iteratively using Lemma \ref{lem:varred}. At each step, we shall use Lemma \ref{lem:lower_bound} to find a $W_j$ variable such that $S$ contains a sub-polynomial of significant mass which is not divisible by $W_j^2$.

	As a first step, using \eqref{eq:Delta} and the definition of $\Delta_1^{(1)}$, for every $j \in [T/2]$, we can write:
\begin{equation}\label{eqn:str}
S(\mb{W}) = \left(W_j - \sum \mb{W}_{\neq j}\right) \cdot S_j^{(1)}(\mb{W}_{\neq j}) + W_j^2 \cdot R_j^{(1)}(\mb{W}) + \Delta_j^{(1)}(\mb{W})
\end{equation}
where $S_i^{(1)}$, $R_i^{(1)}$ and $\Delta_j^{(1)}$ are polynomials of degrees at most $d-1$, $d-2$ and $d$ respectively and $\|\Delta_j^{(1)}\|\leq C_{\max}$. Because $T/2 > d$, using Lemma \ref{lem:lower_bound} and re-indexing, we can assume that the sub-polynomial of $S$ not divisible by $W_1^2$ has $\ell_2$-norm at least $\eta \coloneqq 4^{-2^d}$. 

Now, applying the variable reduction lemma (Lemma \ref{lem:varred})  for every $j \in [2, T/2]$, with $a=1, X=W_1, Y=W_j,$ and $Z = \sum \mb{W}_{\neq 1,j}$, we obtain that there exist polynomials $S^{(2)}_j$, $R^{(2)}_j$ and $\Delta^{(2)}_j$ of degrees $d-2$, $d-3$ and $d-1$ respectively such that
$$S^{(1)}_1(\mb{W}_{\neq 1}) = \left(2W_j - \sum \mb{W}_{\neq 1, j}\right) \cdot S^{(2)}_j(\mb{W}_{\neq 1,j}) + W_j^2 \cdot R^{(2)}_j(\mb{W}_{\neq 1}) + \Delta_j^{(2)}(\mb{W}_{\neq 1})
$$
and $\|\Delta_j^{(2)}\| \leq 20C_{\max}$. Again, by Lemma \ref{lem:lower_bound} and re-indexing, we can ensure that the sub-polynomial of $S_1^{(1)}$ not divisible by $W_2^2$ has $\ell_2$-norm at least $\eta \|S^{(1)}_1\|$. 
 
Applying the variable reduction lemma again with $a=2$, we obtain polynomials $S_j^{(3)}, R_j^{(3)}$ and $\Delta_j^{(3)}$ of degrees $d-3, d-4,$ and $d-2$ respectively such that for any $j \in [3,T/2]$:
$$S^{(2)}_2(\mb{W}_{\neq 1,2}) = \left(3W_j - \sum \mb{W}_{\neq 1, 2, j}\right) \cdot S^{(3)}_j(\mb{W}_{\neq 1,2, j}) + W_j^2 \cdot R^{(3)}_j(\mb{W}_{\neq 1,2}) + \Delta_j^{(3)}(\mb{W}_{\neq 1,2})
$$
and $\|\Delta_j^{(3)}\| \leq 20^2 \cdot 2 \cdot C_{\max}$. Continuing this way, we get that for every $1\leq \ell < j \leq T/2$, there exist polynomials $S^{(\ell)}_j, R^{(\ell)}_j$ and $\Delta^{(\ell)}_j$ of degrees $d-\ell, d-\ell-1$, and $d-\ell+1$ such that:
\begin{eqnarray}
S^{(\ell-1)}_{\ell-1}(\mb{W}_{\neq [\ell-1]}) & = & \left(\ell W_j - \sum \mb{W}_{\neq [\ell-1]\cup \{j\}}\right) \cdot S^{(\ell)}_j(\mb{W}_{\neq [\ell-1] \cup \{j\}}) \nonumber \\
& & + W_j^2\cdot R^{(\ell)}_j(\mb{W}_{\neq [\ell-1]}) + \Delta_j^{(\ell)}(\mb{W}_{\neq [\ell-1]})\label{eqn:genl}
\end{eqnarray}
and $\|\Delta_j^{(\ell)}\| \leq (20\ell)^{\ell-1} C_{\max}$.  Here, $S_0^{(0)} = S$. Moreover, using Lemma \ref{lem:lower_bound}, we can assume that the sub-polynomial of $S_{\ell-1}^{(\ell-1)}$ not divisible by $W_\ell^2$ has $\ell_2$-mass at least $\eta \|S^{(\ell-1)}_{\ell-1}\|$. 

For $\ell = d$, we obtain a linear polynomial $S^{(d-1)}_{d-1}(\mb{W}_{\neq 1, \dots, d-1})$ such that for every $j \in [d, T/2]$, there exists constant $S^{(d)}_j$ and linear polynomial $\Delta^{(d)}_j$ such that:
$$S_{d-1}^{(d-1)}(\mb{W}_{\neq 1, \dots, d-1}) = \left(dW_j - \sum \mb{W}_{\neq 1,\dots,d-1, j}\right) \cdot S^{(d)}_j + \Delta^{(d)}_j(\mb{W}_{\neq 1, \dots, d-1})
$$
Note that $R_j^{(d)} = 0$ because $S^{(d-1)}_{d-1}$ is not divisible by $W_j^2$ being a linear polynomial. 

Applying Lemma \ref{lem:varred} one final time, we get that $|S_d^{(d)}| \leq (40d)^{d} C_{\max}$. On the other hand, we have the following claim:
\begin{claim}
For any $0\leq \ell \leq T/2$, $\|S^{(\ell)}_\ell\| \geq \left(\frac{\eta}{T}\right)^\ell - 2\frac{(20\ell)^\ell C_{\max}}{T}$. 
\end{claim}
\begin{proof}
The proof is by induction. For $\ell=0$, the claim is true because $\|S_0^{(0)}\| = \|S\| = 1$. For the induction, note that by our choice of the index $\ell$ above, the sub-polynomial of  $S^{(\ell-1)}_{\ell-1}$ not divisible by $W_\ell^2$ has $\ell_2$-mass at least $\eta \|S^{(\ell-1)}_{\ell-1}\|$. Moreover, from \eqref{eqn:genl}  and triangle inequality this mass is at most
\[
\|(\ell W_\ell - \sum \mb{W}_{\neq [\ell]})S_\ell^{(\ell)}\| + \|\Delta_\ell^{(\ell)}\|
\] 
So: 
\begin{align*}
\|(\ell W_\ell - \sum \mb{W}_{\neq [\ell]})S_\ell^{(\ell)}\| &\geq \eta \|S^{(\ell-1)}_{\ell-1}\| - \|\Delta_\ell^{(\ell)}\| \\
&\geq \eta \|S^{(\ell-1)}_{\ell-1}\| - (20\ell)^{\ell-1} C_{\max} \\
&\geq \eta^\ell/T^{\ell-1} - 2\eta(20 \ell)^{\ell-1}C_{\max}/T  -(20\ell)^{\ell-1}C_{\max}\\ 
&\geq \eta^\ell/T^{\ell-1}- 2 (20\ell)^{\ell}C_{\max}
\end{align*}
The claim follows by observing $\|(\ell W_\ell - \sum \mb{W}_{\neq [\ell]})S_\ell^{(\ell)}\| \leq T\cdot \|S^{(\ell)}_\ell\|$. 
\end{proof}

Therefore, $|S^{(d)}_d| \geq (\eta/T)^d - 2(20d)^d C_{\max}/T$. But by our choice of $\eta$ and $C_{\max}$,  $(\eta/T)^d - 2(20d)^d C_{\max}/T > (40d)^d C_{\max}$, since $C_{\max}((40d)^d + 2(20d)^d/T) < C_{\max} (80d)^d <  (1/4T)^{2^d} = (\eta/T)^d$. This is a contradiction.

\bibliographystyle{alpha}
\bibliography{Refs-LC-arxiv}

\appendix
\section{Useful Tools and Results}\label{sec-useful-appendix}
\begin{fact}\label{fact-sum-0}
There exists a distribution of random variables $g_1, \dots, g_R$ such
that each $g_i$ is marginally $N(0,1)$, $\E[g_ig_j] = -1/(R-1)$ for
all $i\neq j$, and $\sum_{i=1}^R g_i = 0$.
\end{fact}
\begin{lemma}\label{lem-ortho-transform}
Let ${\bf g} = (g_1 \dots g_R)^{\tn{T}}$ where $\{g_i\}_{i=1}^R$ are 
as given in Fact \ref{fact-sum-0}, and suppose
${\bf x} = (x_1 \dots, x_R)^{\tn{T}}$, ${\bf y} = (y_1 \dots
y_R)^{\tn{T}} \in \R^R$ are orthogonal unit vectors such that
$\langle \mathbf{1}, {\bf x}\rangle = 0$ and  $\langle \mathbf{1}, {\bf
y}\rangle = 0$. Define, $f := \langle {\bf x}, {\bf g}\rangle$ and $h
:= \langle {\bf y}, {\bf g}\rangle$.  Then, $f$ and $g$ are 
independent $N(0, R/(R-1))$ random variables.
\end{lemma}
\begin{proof}
We have,
\begin{eqnarray}
\displaystyle \E[f^2] & = &
\E\left[\left(\sum_{i=1}^Rx_ig_i\right)^2\right] \nonumber \\
& = & \sum_{i=1}^Rx_i^2\E[g_i^2] + \sum_{\substack{i,j\in [R]\\ i\neq
j}}x_ix_j \E[g_ig_j] \nonumber \\
& = & \sum_{i=1}^Rx_i^2 - \left(\frac{1}{R-1}\right)
\sum_{\substack{i,j\in [R]\\ i\neq j}}x_ix_j \nonumber \\
& = & \left(1 + \frac{1}{R-1}\right)\sum_{i=1}^Rx_i^2 -
\left(\frac{1}{R-1}\right)\left(\sum_{i=1}^Rx_i\right)^2 \nonumber \\
& = & \frac{R}{R-1}.
\end{eqnarray}
The same holds for $\E[h^2]$. For the second part of the lemma
observe that,
\begin{eqnarray}
\E[fh] & = & \sum_{i=1}^R\left[x_iy_i\E[g_i^2] + \sum_{\substack{j\in [R] \\ j\neq
i}}\E[g_ig_j]x_iy_j \right] \nonumber \\
& = & \sum_{i=1}^R\left[ x_iy_i - \left(\frac{1}{R-1}\right)\sum_{\substack{j\in [R] \\ j\neq
i}}x_iy_j\right] \nonumber \\
& = & \left(1 + \frac{1}{R-1}\right)\langle {\bf x}, {\bf y}\rangle -
\left(\frac{1}{R-1}\right)\langle {\bf x}, {\bf 1}\rangle\langle {\bf
y}, {\bf 1}\rangle \ = \ 0. 
\end{eqnarray} 
\end{proof}

\begin{fact}[Fact 3.4 in \cite{DOSW11}]\label{fact:hermite}
Let $P: \R^\ell \to \R$ be a degree-$d$ polynomial over independent standard normal variables which has at least one coefficient of magnitude at least $\alpha$. Then, $\|P\|_2 \equiv \sqrt{\E[|P({\bf x})|^2]}$ is at least $\frac{\alpha}{d^d {\ell + d \choose d}}$. 
\end{fact}

\section{Comparing monomial and $\ell_2$-masses}

In this section, we relate the monomial mass of the polynomials with their $\ell_2$-mass under the distribution $\mc{D}$.

\begin{lemma}				\label{lem:coeff_bounds}
	Let $Q(U_1,\ldots,U_T)$ be a polynomial of degree $d \ge 1$. Let $\tilde{Q}(W_{1,1},\ldots,W_{1,T})$ be the polynomial obtained from $Q(U_1,\ldots,U_T)$ by the orthonormal transformation. With $\eta$ and $T = 10d$ chosen as in Section \ref{sec:reduction}, the following bounds hold:
	
	\begin{itemize}
		\item[1.] $\|Q(U_1,\ldots,U_T)\|_2 \le (20dT)^{5d}\|\tilde{Q}(W_{1,1},\ldots,W_{1,T})\|_{\rm mon,2}$		
		\item[2.] If $Q$ depends only on variables $U_2,\ldots,U_T$ then $\|\tilde{Q}(W_{1,1},\ldots,W_{1,T})\|_{\rm mon,2} \le (10dT)^{7d}\|Q(U_2,\ldots,U_T)\|_2$
	\end{itemize}
\end{lemma}

\begin{proof}
	For ease of notation, we shall denote variables $W_{11},\ldots,W_{1T}$ by $W_{1},\ldots,W_{T}$. Let $\mcb{S}_{T,d}$ be the set of all multi-sets on $[T]$ of size at most $d$. Using the fact that ${T\choose d} \le \Big(\frac{Te}{d}\Big)^d \le (eT)^{d}$ we have $|\mcb{S}_{T,d}| \le (10T)^{2d}$\\ 
	
	\noindent{\bf Proof of Part $1.$}:  For the first direction let ${Q}(U_1,\ldots,U_T) = \sum_{S \in \mcb{S}_{T,d}}c_SU_S$, where the monomial $U_S$ is defined as $U_S = \prod_{i \in S}U^{S(i)}_i$. Therefore,
	
	\begin{eqnarray}
	\|Q\|^2_2 &=& \E_{\mc{D}_{\mc{I}}}\bigg[\Big( \sum_{S \in \mcb{S}_{T,d}} c_SU_S\Big)^2\bigg] \\
	&\le& \E_{\mc{D}_{\mc{I}}}\bigg[ \Big(\sum_{S \in \mcb{S}_{T,d}} c^2_S\Big) \Big(\sum_{S \in \mcb{S}_{T,d}} U^2_S\Big)\bigg] \\
	&=&  \|Q(U_1,\ldots,U_T)\|^2_{\rm mon,2} \bigg(\E_{\mc{D}_{\mc{I}}}\bigg[\sum_{S \in \mcb{S}_{T,d}} U^2_S\bigg]\bigg) 	\label{eq:appx1}
	\end{eqnarray}		
	
	For the first term, we claim that 
	
	\begin{eqnarray*}
	\|Q(U_1,\ldots,U_T)\|_{\rm mon,2} &\le& \|Q(U_1,\ldots,U_T)\|_{\rm mon,1} \\
	&\le& ({10T})^{3d} \|\tilde{Q}(W_{1},\ldots,W_{T})\|_{\rm mon,1} \\
	&\le& ({10T})^{4d} \|\tilde{Q}(W_1,\ldots,W_T)\|_{\rm mon,2} 
	\end{eqnarray*}
	
	\noindent where the first inequality follows the fact that $\ell_2$-norm is upper bounded by the $\ell_1$-norm, and the third inequality follows from \emph{Cauchy-Schwarz} and $|\mcb{S}_{T,d}| \le (10T)^{2d}$. The middle inequality can be argued as follows. Consider $U_S = \prod_{i \in S} U^{S(i)}_i $. Then it can be expressed as in terms of $W_{1},\ldots,W_{T}$ as 
	\begin{equation*}
		\prod_{i \in S} \Big(\sum_{l \in [T]} a_{i,l}W_{l}\Big)^{S(i)} 
	\end{equation*}
	
	By construction, the linear transformation $\{U_1,\ldots,U_T\} \mapsto \{W_1,\ldots,W_T\}$ is \emph{orthonormal} (See section \ref{sec:R=1Basis}). Therefore each coefficient satisfies $|a_{i,l}| \le 1$. Furthermore, there can be at most $T^d$ distinct terms in the expansion of $U_S$. Therefore, the total contribution to the coefficient of a fixed monomial from $U_S$ can be at most $|c_S|T^d$. Repeating the argument across all $S \in \mcb{S}_{T,d}$ completes the argument. \\
	
	For upper bounding the expectation term in   \eqref{eq:appx1}, fix a $S \in \mcb{S}_{T,d}$. Then,
	\begin{eqnarray*}
	\E_{\mc{D}_{\mc{I}}} \Big[U^2_S\Big] &=&  \E_{\mc{D}_{\mc{I}}} \bigg[\prod_{i \in S}U^{2S(i)}_i\bigg] \\
	    &=&  \prod_{i \in S}\E_{\mc{D}_{\mc{I}}} \bigg[U^{2S(i)}_i\bigg]  \quad\qquad\qquad\Big(\mbox{Since } U_1,\ldots,U_T \mbox{ are independent}\Big)\\
    	&\le&  \prod_{i \in S \setminus \{1\}}\E_{\mc{D}_{\mc{I}}} \bigg[U^{2S(i)}_i\bigg] \qquad \qquad \Big(\mbox{Since } \eta\sqrt{T} < 1\Big) \\
	    &\overset{1}{\le}&  \prod_{i \in S \setminus \{1\}}(2S(i))! \\
	    &\le&  (2|S|)! 
	\end{eqnarray*}
	where step $1$ follows from the well known fact that for $g \sim N(0,1)$, $\E[g^k] \le k!$ for all $k \in \Z_+$. Therefore, plugging in the upper bounds in   \eqref{eq:appx1} we get
\begin{eqnarray*}
\|Q(U_1,\ldots,U_T)\|^2_{\rm mon,2} \bigg(\E_{\mc{D}_{\mc{I}}}\bigg[\sum_{S \in \mcb{S}_{T,d}} U^2_S\bigg]\bigg) 
\le ({10T})^{10d} (2d)^{(2d)}\|\tilde{Q}({W}_1,\ldots,W_T)\|^2_{\rm mon,2} 
\end{eqnarray*}
	
	\noindent{\bf Proof of Part $2$}: For the second direction, we observe that
	
	\begin{eqnarray}
	\|\tilde{Q}(W_1,\ldots,W_T)\|_{\rm mon,2} &{\le}& \|\tilde{Q}(W_1,\ldots,W_T)\|_{\rm mon,1} \\
	&\overset{1}{\le}& ({10T})^{3d} \|{Q}(U_2,\ldots,U_T)\|_{\rm mon,1} \\
    &\overset{2}{\le}& ({10dT})^{7d} \|{Q}(U_2,\ldots,U_T)\|_2			 		
	\end{eqnarray}
	
	\noindent where inequality $1$ again can be argued similarly to the previous direction (using the fact that $\{W_1,\ldots,W_T\} \mapsto \{U_1,\ldots,U_T\}$ is again an orthonormal linear transformation).  
	
	For step $2$, we write $Q(U_2,\ldots,U_T)$ in the monomial basis of $U$ i.e., $Q(U_2,\ldots,U_T) = \sum_{S} c_S U_S$ and see that 
	
	\begin{eqnarray}
	\bigg\|\sum_{S} c_S U_S\bigg\|_{\rm mon,1} = \sum_{S \in \mcb{S}_{T-1,d}}|c_S| \overset{1}{\le} \sum_{S \in \mcb{S}_{T-1,d}}({6Td})^{2d}\|Q(U_2,\ldots,U_T)\|_2 \le ({10dT})^{4d}\|Q(U_2,\ldots,U_T)\|_2
	\end{eqnarray}
	
	\noindent with step $1$ following from Fact \ref{fact:hermite}, and the last inequality uses the upper bound on $|\mcb{S}_{T,d}|$.

\end{proof}

\section{Comparison inequalities between Norms}

\begin{cl}							\label{cl:compare_masses2}					
	Given polynomials $P_1({\bf W}),P_2({\bf W})$ over variables ${\bf W} = (W_{11},\ldots,W_{1T})$, we have 
$$\|P_1({\bf W})P_2({\bf W})\|_{\rm mon,2} \le \|P_1({\bf W})\|_{\rm mon,1}\|P_2({\bf W})\|_{\rm mon,2}.$$
\end{cl}
\begin{proof}
	Let $P_1({\bf W}) = \sum_{W_S \in \mcb{M}} c_SW_S$. Then,
	\begin{eqnarray*}
	\|P_1({\bf W})P_2({\bf W})\|_{\rm mon,2} &=& \Big\|\sum_{W_S \in \mcb{M}} c_SW_SP_2({\bf W})\Big\|_{\rm mon,2} \\
	&\le& \sum_{W_S \in \mcb{M}}|c_S|\|W_SP_2({\bf W})\|_{\rm mon,2} \\
	&=& \|P_1({\bf W})\|_{\rm mon,1}\|P_2({\bf W})\|_{\rm mon,2}
	\end{eqnarray*}
\end{proof}

\end{document}